\documentclass[twocolumn]{autart}    

\usepackage{graphicx}          
\newtheorem{theorem}{\textbf{Theorem}}
\newtheorem{lemma}{\textbf{Lemma}}
\newtheorem{example}{\textbf{Example}}
\newtheorem{assumption}{\textbf{Assumption}}
\newtheorem{corollary}{\textbf{Corollary}}
\newtheorem{remark}{\textbf{Remark}}
\newtheorem{definition}{\textbf{Definition}}

\newenvironment{proof}{{{\bf Proof:}}}{\hfill $\square$\par}
\usepackage{amsmath}
\usepackage{amssymb}
\usepackage{epstopdf}
\usepackage{float}
\usepackage{xcolor}
\usepackage{varwidth}
\begin{document}
	
	\begin{frontmatter}
		
		\title{Data-Driven Observers Design for Descriptor Systems} 
		
		\thanks[footnoteinfo]{This paper was not presented at any IFAC
			meeting. This work was supported by the National Key Research and Development Program of China under Grant 2025YFA1018800 and the National Natural Science Foundation of China under Grant 62373059. Corresponding authors: Keke Huang and Zhongqi Sun.}
		
		\author[a]{Yuan Zhang}\ead{zhangyuan14@bit.edu.cn},    
		\author[a]{Yu Wang}\ead{3220241197@bit.edu.cn},    
        	\author[b]{Keke Huang}\ead{huangkeke@csu.edu.cn}, 
        	\author[a]{Zhongqi Sun}\ead{zhongqisun@bit.edu.cn}, 
		and \author[c]{Tyrone Fernando}\ead{tyrone.fernando@uwa.edu.au}  
		\address[a]{School of Automation, Beijing Institute of Technology, Beijing, China} 
        		\address[b]{School of Automation,
Central South University
Changsha, China}

        		\address[c]{Department of Electrical, Electronic and Computer Engineering, University of Western Australia (UWA), Crawley, WA 6009, Australia} 


		\begin{abstract}                          
			State estimation constitutes a core task in monitoring, supervision, and control of dynamic systems. This paper proposes a data-driven framework for the design of state observers for descriptor systems. Necessary and sufficient conditions for the existence of a standard state observer are derived purely from data under mild assumptions. When the system is subject to unknown inputs, we further extend the framework to the data-driven design method for full-order unknown input observer (UIO). Notably, for both the standard state observer and the UIO, we establish the mathematical equivalence between the proposed data-driven existence conditions and classical model-based ones. Moreover, the data-driven approach is applied to the design of extended state observers, enabling simultaneous estimation of system states and disturbances via system augmentation. Numerical simulations validate the effectiveness of the proposed methods. 
		\end{abstract}
		
	\end{frontmatter}
	
	\section{Introduction}
	State estimation \cite{chernousko1993state}, a core problem in control theory, aims to reconstruct the internal states of a system that cannot be directly measured, using input and output measurements. The development of asymptotic observer theory can be traced back to the foundational contributions of Luenberger \cite{luenberger1966observers,luenberger1979introduction}, which were principally concerned with systems characterized by linear state-space dynamics. The standard formulation of the problem presupposes that a complete mathematical model of the system is given, and that all input and output signals influencing the system are accessible. However, in practical complex systems—such as power grids, chemical processes, economic networks and so on \cite{gao2007actuator,duan2010analysis,abbaszadeh2012generalized}—obtaining precise analytical models is often challenging. These systems typically involve algebraic constraints coupled with dynamic behaviors, making them more suitably described by descriptor systems. Distinct from usual linear time-invariant (LTI) systems, descriptor systems exhibit unique characteristics, such as impulsive behavior, non-causality, and eigenvalues at infinity, which pose significant challenges to the design of state estimators. The inherent algebraic-differential structure of descriptor systems necessitates specialized approaches to ensure stability and accuracy in state reconstruction, particularly when dealing with inaccessible inputs or external disturbances.
	
	In recent years, several model-based observer design methods have been proposed for descriptor systems. For example, \cite{wang2012observer} introduced a full-order observer design for discrete-time descriptor systems under certain assumptions. \cite{guo2015reduced} focused on discrete-time descriptor systems and developed a reduced-order observer design, later extending it to a class of nonlinear descriptor systems. Based on solutions to the Riccati equation, \cite{shields1996observer} designed a nonlinear observer for a class of continuous-time nonlinear descriptor systems subject to unknown inputs and faults. For linear time-invariant continuous-time descriptor systems, \cite{darouach2002design} investigated the reduced-order observer design using the generalized Sylvester equation. Furthermore, \cite{yang2013nonlinear} established a unified framework for designing both full-order and reduced-order observers for a class of continuous-time nonlinear descriptor systems with time-varying nonlinearities satisfying quadratic constraints.
	
	Although model-based observer design methods have achieved significant progress in the study of descriptor systems \cite{koenig2002design,koenig2006observer,dai1988observers,darouach2010h}, their practical effectiveness relies heavily on the availability of explicit mathematical models. However, obtaining such system matrices is often a challenging and time-consuming task. When dealing with complex physical interactions, the system identification process becomes prohibitively expensive or computationally burdensome \cite{moonen1992subspace,he2022identification}.
	
	To address the challenges where explicit physical models are difficult or costly to obtain, data-driven control methods have emerged as a promising alternative. These approaches eliminate the reliance on explicit mathematical models and instead leverage operational data directly to design controllers or observers. The framework represented by Willems' Fundamental Lemma \cite{van2020willems,schmitz2022willems} demonstrates that, under persistent excitation conditions, the Hankel matrix constructed from input–output data can uniquely characterize the full dynamic behavior of a system. This provides a theoretical foundation for constructing state observers directly from the data.

	
	Despite significant progress in data-driven estimation for non-singular linear systems \cite{turan2021data,disaro2024equivalence,zhang2025data}, considerable gaps remain in extending them to descriptor systems due to the algebraic constraints and inherent non-causality of discrete-time descriptor systems. In this paper, we propose a comprehensive data-driven framework for state estimation in descriptor systems and rigorously establish the equivalence between model-based and data-driven approaches. Our principal contributions are fourfold:
	\begin{enumerate}
		\item For linear discrete-time descriptor systems without unknown inputs, we derive necessary and sufficient conditions for the existence of a state observer based solely on historical data matrices. Specifically, to address the challenges posed by the non-causal nature of descriptor systems, the one-step-ahead future output is explicitly introduced into the observer structure.
		\item For descriptor systems subject to unknown inputs, we develop a data-driven design method for the full-order Unknown Input Observer (UIO), establishing its existence conditions and design framework.
		\item We prove the mathematical equivalence between the proposed data-driven and the classical model-based design approaches. This theoretical analysis encompasses both the standard state observer and the UIO case.
		\item Finally, we apply the proposed data-driven framework to the design of Extended State Observers (ESOs) for LTI systems. By constructing an augmented system, we demonstrate that the proposed method can achieve simultaneous estimation of system states and disturbances without relying on explicit system matrices.
	\end{enumerate}
	Interestingly, we find that compared to the data requirement for the observer design of non-singular systems \cite{turan2021data,disaro2024equivalence}, descriptor systems generally require less data. Additionally, it turns out that descriptor systems and standard LTI systems share the uniform data-based observer design procedure.
	This is appealing, as the proposed data-based procedure 
	is applicable without identifying the system singularity, which may be challenging~\cite{he2022identification}.  
	
	The subsequent sections are organized as follows: Section~\ref{sec: pre} reviews necessary preliminaries on descriptor systems. Section~\ref{sec: full order} presents the data-driven state observer design and demonstrates the theoretical equivalence with model-based methods. Section~\ref{sec: uio} addresses observers with unknown inputs, detailing the design of full-order UIO. The application of the proposed framework to ESO design is provided in Section~\ref{sec: eso}. Finally, Section~\ref{sec: simulation} validates the effectiveness of the proposed methods through numerical simulations.
	
	Notations: $\mathbb{R}$ and $\mathbb{Z}_{0}$ denote the set of real numbers and the natural numbers with zero. The identity matrix in $\mathbb{R}^{n \times n}$ and the zero matrix in $\mathbb{R}^{n \times n}$ are denoted by $I_n$ and $0_{n\times n}$, respectively. For a matrix $A\in \mathbb{R}^{m\times n}$, we denote by $Im(A)$ and $Ker(A)$ the image and the right kernel of $A$. In addition, we denote by $A^{\dagger} \in \mathbb{R}^{n\times m}$ its Moore-Penrose inverse and ${rk}(A)$ its rank. When $m=n$, $det(A)$ denotes the determinant of $A$. For a set of matrices $A_1,...,A_n$ with the same column numbers, $[A_1^\top,A_2^\top,\cdots,A_n^\top]^\top$ is also denoted by ${\bf col}\{A_1,A_2,\cdots,A_n\}$.
	For a sequence of vectors $\{f(k)\}_{k=0}^{T-1}$, the Hankel matrix of depth $L$ is denoted by $H_{L}(f_{[0,T-1]})$ and defined as\\
	$H_{L}(f_{[0,T-1]}) \triangleq
	\begin{bmatrix}
		f(0) & f(1) & \cdots & f(T-L) \\
		f(1) & f(2) & \cdots & f(T-L+1) \\
		\vdots & \vdots & \ddots & \vdots \\
		f(L-1) & f(L) & \cdots & f(T-1)
	\end{bmatrix}$. If $H_{L}(f_{[0,T-1]})$ has full row rank, we say $\{f(k)\}_{k=0}^{T-1}$ is persistently exciting of order $L$.
	
	\section{Preliminaries  and model-based observers}\label{sec: pre}
	Consider a discrete-time linear descriptor system $\Sigma_1$
	\begin{subequations}
		\begin{align}
			Ex(k+1) &= Ax(k) + Bu(k) \label{2_1a} \\
			y(k)    &= Cx(k) \label{2_1b}
		\end{align}
		\label{2_1}
	\end{subequations}with a consistent initial condition $x(0)$ such that $Ax(0)+Bu(0)\in Im(E)$, where $A,E \in \mathbb{R}^{n\times n}$, $B \in \mathbb{R}^{n \times m}$, $C\in \mathbb{R}^{p\times n}$. We assume that $det(\lambda E - A) \ne 0$ for some $\lambda \in \mathbb{C}$, i.e., regularity of system (\ref{2_1}). Particularly, we are interested in the case where the matrix $E$ is singular, i.e., $rk(E) < n$. We denote the finite
	spectrum set of $(E, A)$ by $\sigma(E,A)$, defined to be $\sigma(E,A)=\{\lambda\in {\mathbb C}: det(\lambda E-A)=0\}$.
	
	
	Since the descriptor system (\ref{2_1}) is regular, there exist invertible matrices $P,S\in \mathbb{R}^{n\times n}$ such that (cf. \cite{Dai1989SingularCS})

\begin{equation}\label{decomposition}
\begin{array}{c}
    SEP = \begin{bmatrix}
        I_{n_1} & 0 \\
        0 & R
    \end{bmatrix}, \quad
    SAP = \begin{bmatrix}
        A_{1} & 0 \\
        0 & I_{n_2}
    \end{bmatrix}, \\
    SB = \begin{bmatrix}
        B_{1} \\
        B_{2}
    \end{bmatrix}, \quad
    CP = \begin{bmatrix}
        C_{1} & C_{2}
    \end{bmatrix},
\end{array}
\end{equation}
	where $R\in \mathbb{R}^{n_2\times n_2}$ is nilpotent with nilpotency index $s\in\mathbb{Z}_{0}$ (i.e., $R^{s-1}\ne 0$ and $R^{s}=0$), and $A_1\in \mathbb{R}^{n_1\times n_1}$, $B_1\in \mathbb{R}^{n_1\times m}$, $B_2\in \mathbb{R}^{n_2\times m}$, $C_1\in \mathbb{R}^{p\times n_1}$, $C_2\in \mathbb{R}^{p\times n_2}$ with $n_1+n_2=n$. Accordingly, $P$ is partitioned as $P = \begin{bmatrix} P_1 & P_2 \end{bmatrix}$, where $P_1 \in \mathbb{R}^{n \times n_1}$ and $P_2 \in \mathbb{R}^{n \times n_2}$. Upon the introduction of the coordinate change $z=P^{-1}x$, system (\ref{2_1}) can be equivalently  written in the following form
	\begin{equation}
		\begin{aligned}
			\begin{bmatrix}
				I_{n_1} & 0 \\
				0 & R
			\end{bmatrix}z(k+1)
			= \begin{bmatrix}
				A_{1} & 0 \\
				0 & I_{n_2}
			\end{bmatrix}z(k)
			+
			\begin{bmatrix}
				B_{1} \\
				B_{2}
			\end{bmatrix}u(k),
		\end{aligned}
		\label{2_2}
	\end{equation}		
	\begin{equation}
		y(k)=\begin{bmatrix}
			C_{1} & C_{2}
		\end{bmatrix}z(k).
		\label{2_3}
	\end{equation}	
	Given an input trajectory $u:\mathbb{Z}_{0}\to\mathbb{R}^{m}$ and an initial value $z_{1} ^{0}\in\mathbb{R}^{n_1}$ there is a unique trajectory such that the state
	$$
	z=
	\begin{bmatrix}
		z_{1}^\top &	z_{2}^\top
	\end{bmatrix}^\top:\mathbb{Z}_{0}\to\mathbb{R}^{n_1+n_2}
	$$
	satisfies $z_{1}(0)=z_{1}^{0}$. This state $z(k)$, $k\in\mathbb{Z}_{0}$, is given by		
	\begin{equation}
		z_{1}(k)=A_{1}^{k}z_{1}(0)+\sum_{m=1}^{k} A_{1}^{k-m}B_{1}u(m-1),
		\label{2_4}
	\end{equation}
	and
	\begin{equation}\label{2_5}
		z_{2}(k) = R^{l-k}z_2(l) - \sum_{m=0}^{l-k-1}R^{m}B_{2}u(k+m), \quad l\ge k.
	\end{equation}
	Here, \eqref{2_4} captures the dynamics of the slow subsystem of system \eqref{2_1}, while \eqref{2_5} the fast subsystem. Since in the fast subsystem, the current state depends on the future inputs (from \eqref{2_5}, $z_2(k)$ is determined by $\{u(k),\cdots,u(k+s-1)\}$),  discrete-time descriptor systems are generically non-causal.
	
	\begin{definition} \cite{Dai1989SingularCS} System \eqref{2_1} is called R-controllable if its slow subsystem \eqref{2_4} is controllable, i.e.,\\ ${rk}([B_1,A_1B_1,\cdots,A^{n_1-1}_1B_1])=n_1$. System \eqref{2_1} is C-controllable if both of its slow and fast subsystems are controllable, i.e., ${rk}([B_1,A_1B_1,\cdots,A^{n_1-1}_1B_1])=n_1$ and ${rk}([B_2,RB_2,\cdots, R^{s-1}B_2])=n_2$. 
		System \eqref{2_1} is called stabilizable if there exists a
		matrix $K\in \mathbb{R}^{m \times n}$, such that $\sigma(E,A + BK)$ is within the unit circle. It is called detectable if its dual system $(E^\top,A^\top,C^\top)$is stabilizable.
	\end{definition}
	\begin{definition} \cite{Dai1989SingularCS}
		System \eqref{2_1} is called observable if its state $x(k)$ at any time point $k$ 
		is uniquely determined by $\{u(i), y(i), i = 0,\cdots,l\}$ where $l$ is sufficiently large.
	\end{definition}
	It is known that system \eqref{2_1} is detectable if and only if~\cite{dai1988observers}\begin{equation}\label{detectability}
		rk\begin{bmatrix} \lambda E - A \\ C \end{bmatrix} = n, \quad \forall \lambda \in \mathbb{C}, \ |\lambda| \geq 1.
	\end{equation}Moreover, system \eqref{2_1} is observable if and only if 
	\begin{equation*}
		rk\begin{bmatrix} \lambda E - A \\ C \end{bmatrix} = n, \, \forall \lambda \in \mathbb{C},  \, {\rm and}\,\, rk\begin{bmatrix}
			E \\ C
		\end{bmatrix} = n.
	\end{equation*}
	For descriptor systems, the singularity of matrix $E$ implies that the state response consists of not only smooth dynamic modes but also algebraic behaviors (or non-causal modes) residing in the fast subsystem.  The following assumption on the observability of the fast subsystem is standard to guarantee the existence of a normal observer \cite{Dai1989SingularCS,dai1988observers,darouach2002design,wang2012observer}.  
	\begin{assumption}\label{assump_obs}
		The fast subsystem of (\ref{2_1}) is observable, i.e. $rk([C_2^\top ,R^\top C_2^\top ,\cdots,(R^{n_2-1})^\top C_2^\top ]^\top) =n_2$, which is equivalent to (see \cite{Dai1989SingularCS})
		\begin{equation}\label{fast_obs}
			rk\begin{bmatrix}
				E \\ C
			\end{bmatrix} = n.
		\end{equation}
	\end{assumption}
	\begin{remark}
		Condition \eqref{fast_obs} is also called the dual normalizability condition \cite[Chap 8]{Dai1989SingularCS}. This condition guarantees that the state $x(k)$ at any time point $k$ is uniquely determined by initial condition $x(0)$ and the former  inputs $u(i)$, together with former measurements $y(i)$, $i = 0,1,...,k$ (i.e., the so-called Y-observability). {How to check this assumption from data will be discussed in Corollary \ref{aump_1_checking} subsequently.}   
	\end{remark}
	
	If condition (\ref{fast_obs}) holds, there exists a full row rank matrix $[T \quad N]$ (not unique) such that
	\begin{equation}\label{system_change}
		TE+NC = I_n.
	\end{equation}
	By using \eqref{system_change}, system \eqref{2_1} is rewritten
	\begin{equation}\label{system_rewrittn}
		x(k+1) = TAx(k) + TBu(k) + Ny(k+1).
	\end{equation}
	Equation \eqref{system_rewrittn} reveals that the state $x(k+1)$ can be explicitly expressed in terms of the current state, input, and one-step-ahead future output. In the non-singular case, i.e., $E=I_n$, we can choose $T=I_n, N=0$ so that \eqref{system_rewrittn} reduces to a standard state-space system.   
	\begin{definition}\label{definition1}
		A linear observer $\hat{\Sigma}$ is designed in the following form:
		\begin{equation}\label{full_order_observer}
			\hat{x}(k+1) = A_{O}\hat{x}(k) + B_{O}^{u}u(k) + B_{O}^{y}y(k) + N_{O}y(k+1),
		\end{equation}
		where $\hat{x}(k)\in \mathbb{R}^n$ denotes the estimate of state $x(k)$. For the observer \eqref{full_order_observer} to be valid, it must satisfy the asymptotic convergence condition:
		\[
		\lim_{k \to \infty} \|x(k) - \hat{x}(k)\| = 0,
		\]
		for any choice of $x(0)$ and input sequence $u(k)$, where $k \in \mathbb{Z}_0$.
	\end{definition}
	
	\begin{remark}In contrast to the observer design for standard LTI systems, the observer \eqref{full_order_observer} explicitly introduces $y(k+1)$ \cite{wang2012observer}. This modification is essential for addressing the inherent non-causal characteristic of descriptor systems.
		Note that the observer $\hat \Sigma$ in \eqref{full_order_observer} can be equivalently written as the following form 
		\begin{subequations}\label{full_order_observer_state}
			\begin{align}
				\zeta(k+1) &= A_O \zeta(k)+B_O^u u(k) + (B_O^y+A_ON_O) y(k),\\
				\hat{x}(k) &= \zeta(k) + N_Oy(k),
			\end{align}
		\end{subequations}
		where $\zeta(k)$ and $\hat{x}(k)$ are the state and output of $\hat{\Sigma}$. 
	\end{remark}

	\section{Data-driven observer design}\label{sec: full order}  
	In this section, we proceed to the data-driven design of observers. An offline dataset comprising input sequence ${u}_d = \{u_d(k)\}_{k=0}^{T+s-1}$, output sequence $y_d = \{y_d(k)\}_{k=0}^{T}$, and state sequence ${x}_d = \{x_d(k)\}_{k=0}^{T}$ is collected from the system \eqref{2_1}.  The subscript $d$ denotes experimentally recorded data. These datasets are partitioned into past $(p)$ and future $(f)$ segments represented as
	$$
	\begin{aligned}
		U_p &= [u_d(0),\dots,u_d(T-1)], \\
		Y_p &= [y_d(0),\dots,y_d(T-1)], \\
		Y_f &= [y_d(1),\dots,y_d(T)], \\
		X_p &= [x_d(0),\dots,x_d(T-1)], \\
		X_f &= [x_d(1),\dots,x_d(T)].
	\end{aligned}
	$$
	\begin{definition}\label{definition:compatible trajectory}
		An (input/output/state) trajectory \((\{u(k)\}_{t \in \mathbb{Z}_{0}}, \{y(k)\}_{t \in \mathbb{Z}_{0}}, \{x(k)\}_{t \in \mathbb{Z}_{0}})\) is compatible with the historical data \((u_{d}, y_{d}, x_{d})\) if
		\begin{equation}\label{compatible trajectory}
			\begin{bmatrix}
				x(k) \\ x(k+1) \\ u(k) \\ y(k) \\ y(k+1)
			\end{bmatrix} \in   Im\left(
			\begin{bmatrix}
				X_{p} \\ X_{f} \\ U_{p} \\ Y_{p} \\ Y_{f}
			\end{bmatrix}\right),
			\quad
			\forall k \in \mathbb{Z}_{0}.
		\end{equation}
		The set of all trajectories compatible with the historical data \((u_{d}, y_{d}, x_{d})\) is denoted by
		\begin{equation}
			\begin{aligned}
				\mathbb{T}_{c}\left(u_{d}, y_{d}, x_{d}\right) \triangleq\{ & \left(\{u(k)\}_{k \in \mathbb{Z}_{0}},\{y(k)\}_{k \in \mathbb{Z}_{0}},\{x(k)\}_{k \in \mathbb{Z}_{0}}\right): \\
				\eqref{compatible trajectory}\,\ holds \} .
			\end{aligned}
		\end{equation}
	\end{definition}
	We now introduce the set of all the (input/output/state) trajectories that can be generated by system \eqref{2_1}:
	$$
	\begin{aligned}
		\mathbb{T}_{\Sigma_1} \triangleq \{\left(\{u(k)\}_{k \in \mathbb{Z}_{0}},\{y(k)\}_{k \in \mathbb{Z}_{0}},\{x(k)\}_{k \in \mathbb{Z}_{0}}\right):\\  \left(\{u(k)\}_{k \in \mathbb{Z}_{0}},\{y(k)\}_{k \in \mathbb{Z}_{0}},\{x(k)\}_{k \in \mathbb{Z}_{0}}\right) \text {satisfies \eqref{2_1}}\}.
	\end{aligned}
	$$
	
	To design a data-driven state observer, we rely on historical data to characterize the system's behavior. A prerequisite for ensuring that  $\mathbb{T}_c$ is fully consistent with $\mathbb{T}_{\Sigma_1}$ is the richness of the collected data.
	Consider the slow subsystem state $z_{1d}(k)$ corresponding to $x_d(k)$ derived from the system decomposition. Let the historical trajectory of this sub-state be denoted as $z_{1d} \triangleq \{z_{1d}(k)\}_{k=0}^{T-1}$.
	\begin{assumption}\label{assump:PE}
		The data $(u_d,x_d,y_d)$ satisfy the condition that the matrix $\begin{bmatrix} H_1(z_{1d[0,T-1]}) \\ H_{s+1}(u_{d[0,T+s-1]}) \end{bmatrix}$ is of full row rank.
	\end{assumption}
	\begin{remark}\label{rem:rank_condition_justification}
		As $z_1(k)$ is an internal state which is difficult to measure, Assumption \ref{assump:PE} is theoretically guaranteed if system \eqref{2_1} is R-controllable and the  $\{u_d(k)\}_{k=0}^{T+s-1}$ is persistently exciting of order $n_1+s+1$ (see \cite[Th.~1]{van2020willems}). Note that $n_1+s<n$ holds for singular systems. Compared to the data requirement for the observer design of usual LTI systems \cite{disaro2024equivalence}, which requires $u_d$ to be presistently exciting of order $n+1$, descriptor systems generally require less data. Compared to common data assumptions made for standard LTI systems \cite{turan2021data,disaro2024equivalence,zhang2025data}, it is worth noting that Assumption \ref{assump:PE} does not guarantee the full row rank of $[U_p^\top,X_p^\top]^\top$.
	\end{remark}
	\begin{lemma}\label{lem:willems}\textbf{\textup{(Equivalence of system and compatible trajectories)}}
		Under Assumption \ref{assump:PE}, the set of all trajectories generated by system \eqref{2_1} exactly coincides with the set of trajectories compatible with the given historical data \( (u_d, y_d, x_d) \), i.e.,
		\begin{equation}\label{eq:equivalence}
			\mathbb{T}_{\Sigma_1} = \mathbb{T}_c(u_d, y_d, x_d).
		\end{equation}
	\end{lemma}
	\begin{proof}
		Let \( \left( \{u(k)\}_{k \in \mathbb{Z}_{0}}, \{y(k)\}_{k \in \mathbb{Z}_{0}}, \{x(k)\}_{k \in \mathbb{Z}_{0}} \right) \in \mathbb{T}_c \). By Definition \ref{definition:compatible trajectory}, there exists a vector  $g_1 \in \mathbb{R}^{T}$  such that
		\begin{equation}\label{eq:compatible} \begin{array}{c} {\bf col}\{ x(k), x(k+1), u(k), y(k),y(k+1)\}\\={\bf col} \{ X_p,X_f,U_p,Y_p,Y_f\}g_1 \end{array}\end{equation}
		Since the historical data is generated by system \eqref{2_1}, it satisfies the system dynamics:
		\begin{equation}\label{eq:historical_data}
			[X_p^\top, X_f^\top, U_p^\top, Y_p^\top, Y_f^\top]^\top=\mathcal{HM}\begin{bmatrix} H_1(z_{1d,[0,T-1]}) \\ H_{s+1}(u_{d,[0,T+s-1]}) \end{bmatrix},
		\end{equation}
		where $\mathcal{H} = \begin{bmatrix}
			[P_1 \quad P_2] & 0 & 0 & 0 & 0 \\
			0 & [P_1 \quad P_2] & 0 & 0 & 0\\
			0 & 0 & I & 0 & 0\\
			0 & 0 & 0 & I & 0\\
			0 & 0 & 0 & 0 & I
		\end{bmatrix}$ and $\mathcal{M} = \begin{bmatrix} 
			I & 0 & 0 & \cdots & 0 & 0 \\
			0 & -B_2 & -RB_2 & \cdots & -R^{s-1}B_2 & 0 \\
			A_1 & B_1 & 0 & \cdots & 0 & 0 \\
			0 & 0 & -B_2 & \cdots & -R^{s-2}B_2 & -R^{s-1}B_2 \\
			0 & I & 0 & \cdots & 0 & 0 \\
			C_1 & - C_2B_2 & -C_2RB_2 & \cdots & -C_2R^{s-1}B_2 & 0 \\
			C_1A_1 & C_1B_1 & -C_2B_2 & \dots & -C_2R^{s-2}B_2 & -C_2R^{s-1}B_2
		\end{bmatrix}$, and the involved matrices are defined in \eqref{decomposition}.
		
		Define $(z_1(k),u(k),\cdots,u(k+s))$ as follows  \begin{equation}\label{eq:assumption_PE}
			\begin{bmatrix} z_1(k) \\ u(k) \\ \vdots \\ u(k+s) \end{bmatrix} =\begin{bmatrix} H_1(z_{1d,[0,T-1]}) \\ H_{s+1}(u_{d,[0,T+s-1]}) \end{bmatrix} g_1.
		\end{equation}
		Substituting \eqref{eq:historical_data} and \eqref{eq:assumption_PE} into \eqref{eq:compatible}, we obtain
		\begin{equation}\label{eq:system}
			\begin{bmatrix}
				x(k)\\ x(k+1) \\ u(k)\\ y(k)\\y(k+1)
			\end{bmatrix} = \mathcal{HM}\begin{bmatrix}
				z_1(k)\\u(k)\\ \vdots \\ u(k+s-1)\\u(k+s)
			\end{bmatrix}.
		\end{equation}
		This shows that the trajectory \( (u(k), y(k),y(k+1), x(k), x(k+1)) \) is generated by system \eqref{2_1} with sub-state $z_1(k)$ and inputs $(u(k),\cdots,u(k+s))$. Hence, $\scalebox{0.9}{$\left( \{u(k)\}, \{y(k)\}, \{x(k)\} \right) \in \mathbb{T}_{\Sigma_1}$}$, proving $\scalebox{0.9}{$\mathbb{T}_{\Sigma_1} \supseteq \mathbb{T}_c(u_d, y_d, x_d)$}$.
		
		Arbitrarily select a trajectory $( \{u(k)\}_{k \in \mathbb{Z}_{0}}$, $\{y(k)\}_{k \in \mathbb{Z}_{0}}$, $\{x(k)\}_{k \in \mathbb{Z}_{0}}) \in \mathbb{T}_{\Sigma_1}$. This trajectory satisfies the system dynamics given by \eqref{eq:system}. Since $\begin{bmatrix} H_1(z_{1d[0,T-1]}) \\ H_{s+1}(u_{d[0,T+s-1]}) \end{bmatrix}$ is of full row rank, there exists a vector $g_1\in \mathbb R^T$ such that \eqref{eq:assumption_PE} holds. Substituting \eqref{eq:historical_data} and \eqref{eq:assumption_PE} into \eqref{eq:system}, it follows directly that the trajectory also satisfies the compatibility condition \eqref{eq:compatible}. This implies
		$\left( \{u(k)\}_{k \in \mathbb{Z}_{0}}, \{y(k)\}_{k \in \mathbb{Z}_{0}}, \{x(k)\}_{k \in \mathbb{Z}_{0}} \right) \in \mathbb{T}_c(u_d, y_d, x_d)$,
		thereby proving the inclusion \(\mathbb{T}_{\Sigma_1} \subseteq \mathbb{T}_c(u_d, y_d, x_d)\).
	\end{proof}
	\begin{remark}[Practical data collection strategy]\label{rem_data_collection}
		Since the structural indices $n_1$ and $s$ are unknown in a purely data-driven setting, one can adopt a conservative strategy based on the system dimension $n$ (noting that $n_1,s \le n$). To ensure the validity of Assumption \ref{assump:PE}, the input data length can be extended to cover the interval $[0, T+n-1]$ for some $T$ (length of the state data), and the input sequence is chosen to be persistently exciting of order $n+1$.
	\end{remark}

	Provided the historical data satisfies the informativity requirements (as formalized in Assumption \ref{assump:PE}), Lemma \ref{lem:willems} shows that the subspace spanned by the history trajectories uniquely characterizes the system's behavior. The following theorem shows the design of observers from these data matrices.
	
	\begin{theorem}[Data-driven observer design]\label{thm:data_driven observer}
		Under Assumptions \ref{assump_obs} and \ref{assump:PE}, there exists a state observer of the form (\ref{full_order_observer}) for system (\ref{2_1}), if there exists matrices $\Sigma_{X_p}$, $\Sigma_{U_p}$, $\Sigma_{Y_p}$, and $\Sigma_{Y_f}$ with compatible dimensions such that \begin{equation}\label{eq:his_into_observer}
			X_f = \begin{bmatrix} \Sigma_{X_p} &  \Sigma_{U_p} & \Sigma_{Y_p} & \Sigma_{Y_f} \end{bmatrix}
			\begin{bmatrix}
				X_p\\U_p\\Y_p\\Y_f
			\end{bmatrix},
		\end{equation}in which $\Sigma_{X_p}$ is Schur stable. 
		Under this condition, a state observer of the form \eqref{full_order_observer} can be constructed with \begin{equation}\label{matrix_match}
			\begin{aligned}
				A_{O} = \Sigma_{X_p},\,\, B_{O}^{u} = \Sigma_{U_p},\,\
				B_{O}^{y} = \Sigma_{Y_p},\,\ N_{O} = \Sigma_{Y_f}.
			\end{aligned}
		\end{equation}
		Furthermore, for any initial estimate $\hat{x}(0)$, the recursive state estimate generated by
		\begin{equation}\label{data-driven-observer}
			\hat{x}(k+1) = \begin{bmatrix} \Sigma_{X_p} & \Sigma_{U_p} & \Sigma_{Y_p} & \Sigma_{Y_f} \end{bmatrix} \begin{bmatrix} \hat{x}(k) \\ u(k) \\ y(k) \\ y(k+1) \end{bmatrix}, \quad k \in {\mathbb Z}_0
		\end{equation}
		asymptotically converges to the true state $x(k+1)$.
	\end{theorem}

	{
		\begin{proof}
			Under Assumption \ref{assump:PE},  Lemma \ref{lem:willems} shows any 
			trajectory $(x(k), u(k), y(k), y(k+1))$ of system \eqref{2_1} satisfies \eqref{eq:compatible}. Consider its sub-equations given as follows
			\begin{equation}\label{solve_compatible trajectory}
				\begin{bmatrix}
					X_p \\ U_p \\ Y_p \\ Y_f
				\end{bmatrix}g_1=
				\begin{bmatrix}
					x(k)\\ u(k) \\ y(k) \\ y(k+1)
				\end{bmatrix}.
			\end{equation}
			The general solution to \eqref{solve_compatible trajectory} is given by  $$g_1 = D_1^{\dagger}
			[ x(k)^\top, u(k)^\top,  y(k)^\top, y(k+1)^\top]^\top+\Omega_1\beta_1,$$
			where $D_1 = [X_p^\top , U_p^\top , Y_p^\top , Y_f^\top]^\top$, $D_1^{\dagger}$ denotes the Moore–Penrose pseudoinverse, $\Omega_1$ is a matrix whose columns form a basis for $Ker(D_1)$, and $\beta_1$ is an arbitrary vector of appropriate dimension. Then, $x(k+1)=X_fg_1$ admits the following state-update representation at each time step $k$:
			\begin{equation}\label{eq:compatible trajectory}
				x(k+1) = X_f g_1 = X_f \left( D_1^{\dagger} \begin{bmatrix} x(k) \\ u(k) \\ y(k) \\ y(k+1) \end{bmatrix} + \Omega_1 \beta \right).
			\end{equation}
			Under Assumption \ref{assump_obs} and based on the equivalent model structure \eqref{system_rewrittn}, $Ker(X_f)\supseteq Ker([X_p^\top , U_p^\top,Y_f^\top]^\top)$, leading to $X_f\Omega_1=0$. Substituting \eqref{eq:his_into_observer} into \eqref{eq:compatible trajectory} yields
			\begin{equation}\label{eq:final}
				x(k+1) = \begin{bmatrix} \Sigma_{X_p} &  \Sigma_{U_p} & \Sigma_{Y_p} & \Sigma_{Y_f} \end{bmatrix}\begin{bmatrix}
					x(k)\\u(k)\\y(k)\\y(k+1)
				\end{bmatrix}.
			\end{equation}Now compare \eqref{eq:final} and \eqref{data-driven-observer} and let the state estimation error be \( e(k) = x(k) - \hat{x}(k) \). Then, $e(k)$ is governed by the autonomous dynamics:
			\[
			e(k+1) = \Sigma_{X_p} e(k), k\in \mathbb{Z}_0.
			\]
			As $\Sigma_{X_p}$ is Schur stable, for any initial error \( e(0) = x(0) - \hat{x}(0) \), the error asymptotically converges to zero, i.e., $\lim_{k \to \infty} e(k) = 0$.
			Consequently, the state estimate \( \hat{x}(k+1) \) generated by \eqref{data-driven-observer} asymptotically converges to the true state \( x(k+1) \) of system \eqref{2_1}.  It is easy to verify that \eqref{data-driven-observer} can be written as the form \eqref{full_order_observer} with parameters given in \eqref{matrix_match}. 
	\end{proof}}

	\begin{remark}\label{remark3}
		From \eqref{eq:his_into_observer}, a general solution for \\
		$[\Sigma_{X_p}, \Sigma_{U_p} ,\Sigma_{Y_p},\Sigma_{Y_f} ]$
		can be expressed as  (\cite[Theorem~2.2]{rao1972generalized}) $\bigl[ \Sigma_{X_p}, \Sigma_{U_p}, \Sigma_{Y_p}, \Sigma_{Y_f} \bigr] = X_f D_1^{\dagger} + K_1(I - D_1D_1^{\dagger})$\label{eq:solve_sigma},
		where $K_1$ is an arbitrary matrix of appropriate dimensions. 
		Partition
		$D_1^{\dagger} = \bigl[ H_{X_p}, H_{U_p}, H_{Y_p}, H_{Y_f} \bigr]$ and $I-D_1D_1^{\dagger} = \bigl[ \Delta_{X_p}, \Delta_{U_p}, \Delta_{Y_p}, \Delta_{Y_f} \bigr]$ in accordance with $[\Sigma_{X_p}, \Sigma_{U_p} ,\Sigma_{Y_p},\Sigma_{Y_f} ]$.
		The necessary and sufficient condition for the existence of a stable $\Sigma_{X_p}$ is the detectability of the pair $ (X_f H_{X_p}, \Delta_{X_p}) $. If $(X_fH_{X_p}, \Delta_{X_p}) $ is observable, one can select a suitable $K_1$ to place the eigenvalues of $\Sigma_{X_p}$ (poles of the observer) to any desirable positions. This provides a degree of freedom for designing observer convergence.
	\end{remark}
	
	\begin{remark}
		It is noteworthy that the mathematical form of the data-based observer in Theorem \ref{thm:data_driven observer} aligns closely with prior data-driven observer frameworks for standard LTI systems \cite{turan2021data,disaro2024equivalence}. As seen in the next section, this resemblance also holds in the UIO case.  Our analysis reveals a significant unification: the derived data-based observer structure is inherently universal, valid for both singular and nonsingular systems. This finding is particularly appealing, as it allows for observer design based purely on data, bypassing the often challenging step of identifying system singularity \cite{moonen1992subspace,he2022identification}.
	\end{remark}
	
	Next, the following Theorem~\ref{thm:equivalence of model and data-driven} establishes the equivalence between model-based and data-driven observer design frameworks. In particular, it shows that the solvability conditions inferred solely from historical data coincide exactly with those obtained from explicit system representations. To facilitate this equivalence proof, it is mathematically essential to first project the collected data matrices onto the decoupled state coordinates $z_1(k)$ and $z_2(k)$. Consistent with the temporal constructions of $X_p$ and $X_f$, we define the corresponding transformed data matrices as follows:
	\begin{align*}
		Z_{1p} &= \begin{bmatrix} z_{1d}(0) & z_{1d}(1) & \cdots & z_{1d}(T-1) \end{bmatrix}, \\
		Z_{1f} &= \begin{bmatrix} z_{1d}(1) & z_{1d}(2) & \cdots & z_{1d}(T) \end{bmatrix}, \\
		Z_{2p} &= \begin{bmatrix} z_{2d}(0) & z_{2d}(1) & \cdots & z_{2d}(T-1) \end{bmatrix}, \\
		Z_{2f} &= \begin{bmatrix} z_{2d}(1) & z_{2d}(2) & \cdots & z_{2d}(T) \end{bmatrix}.
	\end{align*}
	\begin{theorem}\label{thm:equivalence of model and data-driven}\textbf{\textup{(Equivalence of model and data-driven conditions)}}
		Under Assumptions \ref{assump_obs} and \ref{assump:PE}, the following statements are equivalent:\\
		(i) There exists a state observer of the form (\ref{full_order_observer}) for system (\ref{2_1}).\\
		(ii) There exist matrices $\Sigma_{X_p}$, $\Sigma_{U_p}$, $\Sigma_{Y_p}$, and $\Sigma_{Y_f}$ such that (\ref{eq:his_into_observer}) holds and $\Sigma_{X_p}$ is Schur stable.\\
		(iii) \begin{equation}\label{rank_equ}
			rk\begin{bmatrix}
				\lambda X_p-X_f\\U_p\\Y_p\\Y_f
			\end{bmatrix} = rk\begin{bmatrix}
				X_p\\U_p\\Y_f
			\end{bmatrix},
			\forall \lambda \in \mathbb{C},\,\  |\lambda| \geq 1.
		\end{equation}   
		(iv) System \eqref{2_1} is detectable.  
\end{theorem}
\begin{proof}
	$(i)\Rightarrow(ii)$:
	If there exists a state observer of the form (\ref{full_order_observer}), 
	clearly the historical data $(u_d,{x}_d,y_d)$ should satisfy \eqref{full_order_observer} with $\hat x(k)=x_d(k)$, $k\in \mathbb Z_0$, i.e., with zero estimation error. Then $(ii)$ holds with $\Sigma_{X_p}=A_O$, $\Sigma_{U_p}=B_O^u$, $\Sigma_{Y_p}=B_O^y$, and $\Sigma_{Y_f}=N_O$, in which $A_O$ is Schur stable. 
	
	
	$(ii)\Rightarrow(iii)$: To prove $(iii)$, we preliminarily show that
	\[
	\begin{aligned}
		rk\left[\begin{array}{l}X_p \\ U_p \\ Y_p \\ Y_f\end{array}\right]&=rk\left(\left[\begin{array}{ccc}I & 0 & 0 \\ 0 & I & 0 \\ C & 0 & 0 \\ 0 & 0 & I\end{array}\right]\left[\begin{array}{l}X_p \\ U_p \\ Y_f\end{array}\right]\right)
		= rk\left[\begin{array}{l}X_p \\ U_p \\ Y_f\end{array}\right]. 
	\end{aligned}
	\]
	By exploiting condition $(ii)$, we obtain
	\[\begin{aligned}
		&rk\left[\begin{array}{c}\lambda X_p-X_f \\ U_p \\ Y_p \\ Y_f\end{array}\right] \\
		& =rk\left(\left[\begin{array}{cccc}\lambda I-\Sigma_{X_p} & -\Sigma_{U_p} & -\Sigma_{Y_p} & -\Sigma_{Y_f} \\ 0 & I & 0 & 0 \\ 0 & 0 & I & 0 \\ 0 & 0 & 0 & I\end{array}\right]\left[\begin{array}{l}X_p \\ U_p \\ Y_p \\ Y_f\end{array}\right]\right).
	\end{aligned}\]
	Since $\Sigma_{X_p}$ is Schur, it follows that the rank of the matrix on the left coincides with the rank of $\begin{bmatrix}
		X_p^\top & U_p^\top & Y_p^\top & Y_f^\top
	\end{bmatrix}^\top$ for every $\lambda\in\mathbb{C}$ with $|\lambda|\ge 1$.     Therefore, condition $(iii)$ can be derived directly from condition $(ii)$.
	
	$(iii)\Rightarrow(iv)$:
	First, we evaluate the rank of the left-hand data matrix. We preliminarily show that
	\begin{align*}
		rk \begin{bmatrix} \lambda X_p - X_f \\ U_p \\ Y_p \\ Y_f \end{bmatrix} 
		&= rk \left( \begin{bmatrix} I & 0 & 0 \\ 0 & I & 0 \\ 0 & 0 & I \\ -C & 0 & \lambda I \end{bmatrix} \begin{bmatrix} \lambda X_p - X_f \\ U_p \\ Y_p \end{bmatrix} \right) \\
		&= rk \begin{bmatrix} \lambda X_p - X_f \\ U_p \\ Y_p \end{bmatrix}.
	\end{align*}
	Introducing the non-singular transformation matrix $[P_1 \quad P_2]$ and substituting the system dynamics \eqref{2_4} and \eqref{2_5}, it can be deduced that
	\begin{align*}
		&\scalebox{0.75}{$rk \begin{bmatrix} \lambda X_p - X_f \\ U_p \\ Y_p \end{bmatrix} = rk \left( \begin{bmatrix} [P_1 \ P_2] & 0 & 0\\ 0 & I & 0 \\ 0 & 0 & I \end{bmatrix} \begin{bmatrix} \lambda Z_{1p} - Z_{1f} \\ \lambda Z_{2p} - Z_{2f} \\ U_p \\ Y_p \end{bmatrix} \right) = rk \begin{bmatrix} \lambda Z_{1p} - Z_{1f} \\ \lambda Z_{2p} - Z_{2f} \\ U_p \\ Y_p \end{bmatrix}$}\\
		&= rk \scalebox{0.7}{$
			\begin{bmatrix}
				\lambda I - A_1 & -B_1 & 0 & \dots & 0 & 0 \\
				0 & -\lambda B_2 & -(\lambda R-I)B_2 & \dots & -(\lambda R-I)R^{s-2}B_2 & -(\lambda R-I)R^{s-1}B_2 \\
				0 & I_m & 0 & \dots & 0 & 0\\
				C_1 & -C_2B_2 & -C_2RB_2 & \dots & -C_2R^{s-1}B_2 & 0
			\end{bmatrix}
			$} \\
		&= \scalebox{0.75}{$m + rk \begin{bmatrix} \lambda I - A_1 & 0 \\ 0 & -(\lambda R-I)Q \\ C_1 & -C_2RQ \end{bmatrix}$},
	\end{align*}
	where $Q = [B_2 \quad RB_2 \quad \dots \quad R^{s-1}B_2]$ and $R^s=0$ has been used. By exploiting the following algebraic identity:
	\begin{equation*}
		\scalebox{0.9}{$\begin{bmatrix} \lambda I - A_1 & 0 \\ 0 & -(\lambda R-I)Q \\ C_1 & -C_2RQ \end{bmatrix} =  \begin{bmatrix} I_{n_1} & 0 & 0 \\ 0 & -(\lambda R-I) & 0 \\ 0 & -C_2R & I_p \end{bmatrix}\begin{bmatrix} \lambda I - A_1 & 0 \\ 0 & Q \\ C_1 & 0 \end{bmatrix}$},
	\end{equation*}
	one obtains the rank of the left-hand matrix in \eqref{rank_equ} as
	\begin{equation} \label{eq:left_rank}
		rk \begin{bmatrix} \lambda X_p - X_f \\ U_p \\ Y_p \\ Y_f \end{bmatrix} = m + rk \begin{bmatrix} \lambda I - A_1 \\ C_1 \end{bmatrix} + rk(Q).
	\end{equation}
	Next, we consider the right-hand data matrix. In view of \eqref{eq:historical_data}, we have 
	
	
	\begin{equation}\label{relation}
		\scalebox{0.8}{$	\begin{bmatrix} Z_{1p} \\ Z_{2p} \\ U_p \\ Y_f \end{bmatrix}\!\!=\!\!  \begin{bmatrix} 
				I_{n_1} & 0 &  \dots & 0 & 0\\ 
				0 & -B_2 &  \dots & -R^{s-1}B_2 & 0 \\ 
				0 & I_m &  \dots & 0 & 0\\ 
				C_1A_1 & C_1B_1 &  \dots & -C_2R^{s-2}B_2 & -C_2R^{s-1}B_2 	\end{bmatrix}\begin{bmatrix} H_1(z_{1d[0,T-1]}) \\ H_{s+1}(u_{d[0,T+s-1]}) \end{bmatrix}.$}
	\end{equation}
	
	Through some elementary operations, it holds that
	\begin{equation}
		\begin{aligned} \label{key-eq}
			&rk \begin{bmatrix} X_p \\ U_p \\ Y_f \end{bmatrix} = rk \begin{bmatrix} Z_{1p} \\ Z_{2p} \\ U_p \\ Y_f \end{bmatrix} \\
			&={\tiny{ rk \begin{bmatrix} 
						I_{n_1} & 0 & 0 & \dots & 0 & 0\\ 
						0 & -B_2 & -RB_2 & \dots & -R^{s-1}B_2 & 0 \\ 
						0 & I_m & 0 & \dots & 0 & 0\\ 
						C_1A_1 & C_1B_1 & -C_2B_2 & \dots & -C_2R^{s-2}B_2 & -C_2R^{s-1}B_2 
			\end{bmatrix}}} \\
			&= n_1 + m + rk\left( \begin{bmatrix} R \\ C_2 \end{bmatrix} Q \right),
		\end{aligned}
	\end{equation}
	where the second equality is due to \eqref{relation} and Assumption \ref{assump:PE}.
	Recall that by Assumption \ref{assump_obs}, $\bigl[E^\top, C^\top \bigr]^\top$ has full column rank, implying that $\bigl[R^\top, C_2^\top \bigr]^\top$ is of full column rank. Consequently, pre-multiplying by this full column rank matrix does not alter the rank of $Q$, leading to:
	\begin{equation} \label{eq:right_rank}
		rk \bigl[ X_p^\top, U_p^\top, Y_f^\top \bigr]^\top= n_1 + m + rk(Q).
	\end{equation}
	Substituting \eqref{eq:left_rank} and \eqref{eq:right_rank} into the equality of condition $(iii)$ yields:
	\begin{equation}
		m + rk \begin{bmatrix} \lambda I - A_1 \\ C_1 \end{bmatrix} + rk(Q) = n_1 + m + rk(Q).
	\end{equation}We arrive at the detectability condition for the slow subsystem:
	\begin{equation*}
		rk \begin{bmatrix} \lambda I - A_1 \\ C_1 \end{bmatrix} = n_1.
	\end{equation*}
	Since the fast subsystem matrix $(\lambda R-I)$ is inherently non-singular for any $\lambda \in \mathbb{C}$, it can be concluded that the PBH detectability matrix for the full-order system satisfies:
	\begin{equation}
		rk \begin{bmatrix} \lambda I - A_1 & 0 \\ 0 & \lambda R - I \\ C_1 & C_2 \end{bmatrix}= rk\begin{bmatrix}
			\lambda E-A\\C
		\end{bmatrix} = n_1 + n_2 = n.
	\end{equation}
	This completes the proof of $(iii) \Rightarrow (iv)$.
	
	$(iv)\Rightarrow(i)$:
	If condition $(iv)$ is satisfied, the existence of a state observer with the structure defined in (\ref{full_order_observer}) is guaranteed by Theorem 4.1 in \cite{dai1988observers}.
\end{proof}
\begin{remark}
	Since neither $[X_p^\top,U^\top_p,Y^\top_f]^\top$ nor $[X_p^\top,U^\top_p]^\top$ is of full row rank, the proof of the implication from (iii) to (iv) in the above theorem is considerably more involved than the corresponding proof for non-singular systems in \cite{disaro2024equivalence}. Moreover, as $Ker([U_p^\top,X_p^\top]^\top)\subseteq Ker(Y_f)$ does {\emph{not}} hold, $Y_f$ cannot be omitted from the right-hand side of condition \eqref{rank_equ}. This distinguishes condition (iii) from the corresponding data-based condition for non-singular systems presented in \cite[Theorem~9]{disaro2024equivalence}.
\end{remark}

Note that the observability of system \eqref{2_1} is sufficient for Assumption \ref{assump_obs} and the detectability condition \eqref{detectability}. This means existing data-based sufficient conditions for the observability of system \eqref{2_1} (see \cite{wang2025data,wang2026data}) can be used as existence conditions for the data-based observer \eqref{data-driven-observer}. In what follows, we provide a relaxed data-based condition for the satisfaction of Assumption \ref{assump_obs}, with which we can directly verify this core assumption from data without system model information. 
\begin{corollary}[Data-based test for Assumption \ref{assump_obs}] \label{aump_1_checking}
	Condition \eqref{fast_obs} holds if
	\begin{equation}\label{assump1-condition}
		rk([X_p^\top,U_p^\top,Y_f^\top]^\top)=n+m.
	\end{equation}
	Moreover, if system \eqref{2_1} is C-controllable and the data satisfies Assumption \ref{assump:PE}, then condition \eqref{fast_obs} holds if and only if \eqref{assump1-condition} holds.
\end{corollary}
\begin{proof}
	From \eqref{relation}, to make $rk([X_p^\top,U_p^\top,Y_f^\top]^\top)=n+m$, it must hold that $rk([R^\top,C_2^\top]^\top)=n_2$, which leads to condition \eqref{fast_obs}. If system \eqref{2_1} is C-controllable and Assumption \ref{assump:PE} holds, then from \eqref{key-eq}, $rk([X_p^\top,U_p^\top,Y_f^\top]^\top)=n_1+m+rk([R^\top,C_2^\top]^\top)$, noting that $Q\doteq [B_2,\cdots, R^{s-1}B_2]$ has full row rank therein. This leads to the necessity of condition \eqref{assump1-condition} for $rk([R^\top,C_2^\top]^\top)=n_2$. 
\end{proof}

\section{Observer design with unknown inputs}\label{sec: uio}
This section extends the proposed data-driven framework to descriptor systems affected by unknown inputs. Consider the linear discrete-time descriptor system
\begin{subequations}\label{3_1}
	\begin{align}
		Ex(k+1) &= Ax(k)+Bu(k)+F\eta(k), \label{3_1a}\\
		y(k) &= Cx(k), \label{3_1b}
	\end{align}
\end{subequations}
where $\eta(k)\in\mathbb{R}^q$ denotes the unknown input. Without loss of generality, we assume that $F\in\mathbb{R}^{n\times q}$ has full column rank, i.e., $rk(F)=q$. Consistent with the decomposition introduced previously, $F$ is conformably partitioned as $SF=[F_1^\top,F_2^\top]^\top$, where $F_1\in\mathbb{R}^{n_1\times q}$ and $F_2\in\mathbb{R}^{n_2\times q}$.

To establish the data-driven representation for system \eqref{3_1}, we treat $\eta(k)$ as an unmeasured component of the system input. Define the extended input vector $\tilde{u}(k)$ as $ \tilde{u}(k) \triangleq [u(k)^\top,\eta(k)^\top]^\top \in \mathbb{R}^{m+q}$. Suppose that the historical data sequences $\{u_d(k), y_d(k), x_d(k)\}$ are collected from system \eqref{3_1} subject to the unknown input sequence $\{\eta_d(k)\}$. Accordingly, we introduce the required data informativity assumption for the system subject to unknown inputs.
\begin{assumption}\label{assump:PE_uio}
	The matrix
	$ \begin{bmatrix} H_1(z_{1d,[0,T-1]}) \\ H_{s+1}(\tilde{u}_{d,[0,T+s-1]}) \end{bmatrix} =
	\begin{bmatrix} H_1(z_{1d,[0,T-1]}) \\ H_{s+1}(u_{d,[0,T+s-1]}) \\ H_{s+1}(\eta_{d,[0,T+s-1]}) \end{bmatrix}$
	is of full row rank, i.e., with rank $n_1+(m+q)(s+1)$.
\end{assumption}
Analogous to Remark \ref{rem:rank_condition_justification}, even though $\eta(k)$ and $z_1(k)$ are usually unmeasurable, this assumption is guaranteed if system \eqref{3_1} with the augmented input matrix $[B,F]$ is R-controllable and $\{\tilde u_d(k)\}$ is persistently exciting of order $n_1+s+1$.  The latter can be satisfied with high probability if $\{\tilde u_d(k)\}$ is sufficiently random and long.  Under this assumption, Lemma \ref{lem:willems} can be generalized to the case with unknown inputs.
We focus on the set of all feasible trajectories generated by system \eqref{3_1}. Distinct from the standard case in Section~\ref{sec: full order}, we define the set of all {\emph{ compatible}} (input/output/state) trajectories that are generated by system \eqref{3_1} as
\[
\begin{aligned}
	\mathbb{T}_{\Sigma_2} \triangleq \Big\{ & \left(\{u(k)\}_{k\in\mathbb{Z}_0},\{y(k)\}_{k\in\mathbb{Z}_0},\{x(k)\}_{k\in\mathbb{Z}_0}\right) : \\
	& \exists \{\eta(k)\}_{k\in\mathbb{Z}_0} \text{ s.t. }
	\big(\{u(k)\}_{k\in\mathbb{Z}_0}, \{x(k)\}_{k\in\mathbb{Z}_0}, \\
	& \{y(k)\}_{k\in\mathbb{Z}_0}, \{\eta(k)\}_{k\in\mathbb{Z}_0}\big) \text{ satisfies \eqref{3_1}} \Big\}.
\end{aligned}
\]
\begin{lemma}\label{lem:willems_uio_set_equivalence}\textbf{\textup{(Trajectory equivalence with unknown inputs)}}
	Under Assumption \ref{assump:PE_uio}, the set of all trajectories generated by system \eqref{3_1} exactly coincides with the set of trajectories compatible with the historical data collected under unknown inputs (i.e., $\mathbb{T}_{c}(u_d, y_d, x_d)$ as per Definition \ref{definition:compatible trajectory}):
	\begin{equation}\label{eq:equivalence_uio}
		\mathbb{T}_{\Sigma_2} = \mathbb{T}_{c}(u_d, y_d, x_d).
	\end{equation}
\end{lemma}
\begin{proof}
	Let $\left( \{u(k)\}_{k \in \mathbb{Z}_{0}}, \{y(k)\}_{k \in \mathbb{Z}_{0}}, \{x(k)\}_{k \in \mathbb{Z}_{0}} \right) \in \mathbb{T}_c$. By Definition \ref{definition:compatible trajectory}, there exists a coefficient vector $g_2 \in \mathbb{R}^{T}$ such that
	\begin{equation}\label{eq:compatible_uio} \begin{array}{c} {\bf col}\{ x(k), x(k+1), u(k), y(k),y(k+1)\}\\={\bf col} \{ X_p,X_f,U_p,Y_p,Y_f\}g_2 \end{array}\end{equation}
	Since the historical data is inherently generated by the descriptor system \eqref{3_1} under certain unknown inputs, it satisfies the system dynamics over the collected horizon:
	\begin{equation}\label{eq:historical_data_uio}
		[X_p^\top, X_f^\top, U_p^\top, Y_p^\top, Y_f^\top]^\top = \mathcal{H}\bar{\mathcal{M}}
		\begin{bmatrix}
			H_1(z_{1d,[0,T-1]}) \\ H_{s+1}(\tilde{u}_{d,[0,T+s-1]})
		\end{bmatrix},
	\end{equation}
	where $\bar {\mathcal M}$ is obtained from $\mathcal M$ in the proof of Lemma \ref{lem:willems} by changing $B_2$ to the augmented input matrix $[B_2,F_2]$ and $\mathcal H$ is defined after \eqref{eq:historical_data}.
	
	Define a state-input sequence $(z_1(k), u(k), \eta(k), \cdots, u(k+s), \eta(k+s))$ as the linear combination of the historical Hankel matrices using the same vector $g_2$:
	{\small{ \begin{equation}\label{eq:assumption_PE_uio_define}
				{\bf{col}}\{z_1(k), u(k), \cdots, u(k+s), \eta(k+s)\}\!=\! \begin{bmatrix} H_1(z_{1d,[0,T-1]}) \\ H_{s+1}(\tilde{u}_{d,[0,T+s-1]}) \end{bmatrix} g_2,
	\end{equation}}}
	Substituting \eqref{eq:historical_data_uio} and \eqref{eq:assumption_PE_uio_define} into \eqref{eq:compatible_uio}, we obtain
	\begin{equation}\label{eq:system_uio}
		\begin{bmatrix}
			x(k)\\ x(k+1) \\ u(k)\\ y(k)\\y(k+1)
		\end{bmatrix} \!=\! \mathcal{H}\bar{\mathcal{M}}\begin{bmatrix} z_1(k) \\ u(k) \\ \eta(k) \\ \vdots \\ u(k+s)\\ \eta(k+s) \end{bmatrix},
	\end{equation}
	which confirms that the trajectory $\left( \{u(k)\}, \{y(k)\}, \{x(k)\} \right)$ is generated by system \eqref{3_1}. Hence, the trajectory belongs to $\mathbb{T}_{\Sigma_2}$, proving the inclusion $\mathbb{T}_c(u_d, y_d, x_d) \subseteq \mathbb{T}_{\Sigma_2}$.
	
	Conversely, arbitrarily select a trajectory $( \{u(k)\}_{k \in \mathbb{Z}_{0}},$\\$\{y(k)\}_{k \in \mathbb{Z}_{0}}, \{x(k)\}_{k \in \mathbb{Z}_{0}}) \in \mathbb{T}_{\Sigma_2}$. This trajectory intrinsically satisfies the system dynamics formulated in \eqref{eq:system_uio} for some $\{\eta(k)\}_{k\in \mathbb{Z}_0}$. Under Assumption \ref{assump:PE_uio}, there exists a vector $g_2 \in \mathbb{R}^T$ such that \eqref{eq:assumption_PE_uio_define} holds. Substituting \eqref{eq:historical_data_uio} and \eqref{eq:assumption_PE_uio_define} backwards into \eqref{eq:system_uio}, it follows directly that the trajectory satisfies the data compatibility condition \eqref{eq:compatible_uio}. This implies the reverse inclusion $\mathbb{T}_{\Sigma_2} \subseteq \mathbb{T}_c(u_d, y_d, x_d)$.
	Consequently, $\mathbb{T}_{\Sigma_2} = \mathbb{T}_c(u_d, y_d, x_d)$ is established.
\end{proof}

According to Lemma \ref{lem:willems_uio_set_equivalence}, given any compatible trajectory  $[
x(k)^\top, x(k+1)^\top, u(k)^\top, y(k)^\top, y(k+1)^\top
]^\top$ of system \eqref{3_1}, there exists a vector $g_2 \in \mathbb{R}^{T}$ such that
\begin{equation}\label{eq:data_equation_g2}
	\begin{array}{c} {\bf col}\{ x(k), x(k+1), u(k), y(k),y(k+1)\}\\={\bf col} \{ X_p,X_f,U_p,Y_p,Y_f\}g_2 \end{array}
\end{equation}
To predict $x(k+1)$ via $x(k+1) = X_f g_2$, we first determine the coefficient vector $g_2$ by solving the remaining equations, whose general solution is given by
\begin{equation}\label{eq:solve_g2_2}
	g_2 = D_2^{\dagger}\begin{bmatrix}
		x(k)\\u(k)\\y(k)\\y(k+1)
	\end{bmatrix}+\Omega_2\beta_2,
\end{equation}
where $D_2 = [ X_p^\top, U_p^\top, Y_p^\top, Y_f^\top ]^\top$, $D_2^{\dagger}$ denotes the Moore-Penrose pseudoinverse of matrix $D_2$, the columns of matrix $\Omega_2$ form a basis for $Ker(D_2)$, and $\beta_2$ is an arbitrary matrix of appropriate dimensions.

\begin{lemma}\label{lem:compatible trajectory}\textbf{\textup{(Conditions for compatible trajectories)}} Under Assumption \ref{assump:PE_uio},
	a dynamical system of the form \eqref{full_order_observer} can generate all compatible input-output-state trajectories $(\{u(k)\}_{k \in \mathbb{Z}_{0}}, \{y(k)\}_{k \in \mathbb{Z}_{0}}, \{x(k)\}_{k \in \mathbb{Z}_{0}})$ realizable by system \eqref{3_1} if and only if the following kernel inclusion condition holds:
	\begin{equation}\label{kernel inclusion}
		Ker(X_f) \supseteq Ker\left(D_2\right).
	\end{equation}
	Equivalently, this condition guarantees the existence of the coefficient matrices $(\bar{\Sigma}_{X_p}, \bar{\Sigma}_{U_p}, \bar{\Sigma}_{Y_p}, \bar{\Sigma}_{Y_f})$ satisfying:
	\begin{equation}\label{eq:his_into_observer2}
		X_f = [\bar{\Sigma}_{X_p}, \bar{\Sigma}_{U_p}, \bar{\Sigma}_{Y_p}, \bar{\Sigma}_{Y_f}][ X_p^\top, U_p^\top, Y_p^\top, Y_f^\top ]^\top.
	\end{equation}
\end{lemma}
\begin{proof}
	$(\Longleftarrow)$:
	If the kernel inclusion condition \eqref{kernel inclusion} holds, then for any vector $g_2$  satisfying \eqref{eq:solve_g2_2}, the state \( x(k+1) \) is uniquely determined by $X_fg_2$. Consequently, despite the presence of unknown inputs \( \eta(k) \), the state \( x(k+1) \) for any compatible trajectory can be expressed as
	\begin{equation}\label{data-driven-dynamics_2}
		x(k+1) = X_fD_2^{\dagger} [x(k)^\top, u(k)^\top, y(k)^\top, y(k+1)^\top ]^\top.
	\end{equation}
	
	Substituting \eqref{eq:his_into_observer2} into \eqref{data-driven-dynamics_2}, it follows
	\begin{equation}\label{eq:final2}
		x(k+1) = \begin{bmatrix} \bar{\Sigma}_{X_p} & \bar{\Sigma}_{U_p} & \bar{\Sigma}_{Y_p} & \bar{\Sigma}_{Y_f} \end{bmatrix}\begin{bmatrix}
			x(k)\\u(k)\\y(k)\\y(k+1)
		\end{bmatrix}.
	\end{equation}
	Define the observer matrices as $A_{O} = \bar{\Sigma}_{X_p},  B_{O}^{u} = \bar{\Sigma}_{U_p},  B_{O}^{y} = \bar{\Sigma}_{Y_p}, $ and $N_{O} = \bar{\Sigma}_{Y_f}$. The relationship among the elements of the tuple $(x(k), x(k+1), u(k), y(k), y(k+1))$, governed by the data-driven dynamics \eqref{eq:final2}, is equivalent to the state-space structure of observer \eqref{full_order_observer} under this matrix assignment. Consequently, system \eqref{full_order_observer} parametrized by these matrices generates all compatible trajectories of the original system \eqref{3_1}. 
	
	$(\Longrightarrow)$:
	Since the system \eqref{full_order_observer} is capable of generating all trajectories compatible with the historical data of system \eqref{3_1},  \((X_p, X_f, U_p, Y_p, Y_f)\) inherently constitutes a valid trajectory of \eqref{full_order_observer}. By substituting this trajectory into the state-update equation \eqref{full_order_observer}, the kernel inclusion condition \eqref{kernel inclusion} is directly satisfied. This completes the proof.
\end{proof}
In the following, we investigate the existence for a UIO and develop a data-driven framework for state estimation in the presence of unknown inputs.
\begin{theorem}[Full-order data-driven UIO]\label{thm:full_order_data-driven_UIO}
	Suppose that Assumption \ref{assump:PE_uio} holds, and system \eqref{3_1} is subject to unknown inputs. The observer matrices of \eqref{full_order_observer} are then constructed as
	\begin{equation}\label{full_order data-driven UIO matrix}
		\begin{split}
			&A_{O} = \bar{\Sigma}_{X_p},\quad B_{O}^{u} = \bar{\Sigma}_{U_p},\\
			&B_{O}^{y} = \bar{\Sigma}_{Y_p},\quad N_{O} = \bar{\Sigma}_{Y_f},
		\end{split}
	\end{equation}
	which are defined in \eqref{eq:his_into_observer2}.
	Then, a full-order UIO of the form \eqref{full_order_observer} exists if and only if condition \eqref{kernel inclusion} holds and $\bar{\Sigma}_{X_p}$ is Schur stable.
	Consequently, for any initial estimate $\hat{x}(0)$, the state estimate $\hat{x}(k+1)$ generated by the iterative law:
	\begin{equation}\label{data-driven_UIO}
		\hat{x}(k+1) = \begin{bmatrix}\bar{\Sigma}_{X_p} &  \bar{\Sigma}_{U_p} & \bar{\Sigma}_{Y_p} & \bar{\Sigma}_{Y_f} \end{bmatrix} \begin{bmatrix} \hat{x}(k) \\ u(k) \\ y(k) \\ y(k+1) \end{bmatrix}, \quad k \in {\mathbb Z}_0
	\end{equation}
	asymptotically converges to the true state $x(k+1)$.
\end{theorem}
\begin{proof}
	$(\Longleftarrow)$:
	Based on Lemma \ref{lem:willems_uio_set_equivalence} and \ref{lem:compatible trajectory}, system \eqref{full_order_observer} with matrices defined in \eqref{full_order data-driven UIO matrix} is capable of generating all trajectories of system \eqref{3_1}. Compare the system update \eqref{eq:final2} and the estimator update 
	\eqref{data-driven_UIO}. Let the corresponding estimation error be $e(k)=x(k)-\hat x(k)$. We have that 
	\(e(k)\) evolves according to the autonomous dynamics $$e(k+1) = A_{O} e(k).$$ If $A_{O}$ is Schur stable, the estimation error asymptotically converges to zero.
	
	$(\Longrightarrow)$:
	If system \eqref{full_order_observer} qualifies as a UIO for the system \eqref{3_1}, then the kernel inclusion condition \eqref{kernel inclusion} must hold according to Lemma \ref{lem:compatible trajectory}. In addition, the historical data $\{u_d,x_d,y_d\}$ collected from system \eqref{3_1} should satisfy \eqref{full_order_observer} with $\hat x(k)=x_d(k)$, $k\in \mathbb Z_0$, i.e., with zero estimation error. Then \eqref{eq:his_into_observer2} holds with $\bar{\Sigma}_{X_p}=A_O$, $\bar{\Sigma}_{U_p}=B_O^u$, $\bar{\Sigma}_{Y_p}=B_O^y$, and $\bar{\Sigma}_{Y_f}=N_O$, in which $A_O$ is Schur stable.  
\end{proof}
\begin{remark}
	Although both Theorem \ref{thm:data_driven observer} and Theorem \ref{thm:full_order_data-driven_UIO} utilize the identical observer structure \eqref{full_order_observer}, their domains of application and existence conditions are fundamentally different. Theorem \ref{thm:data_driven observer} applies to systems without unknown inputs, requiring only the Schur stability of $A_{O}$. In contrast, Theorem \ref{thm:full_order_data-driven_UIO} addresses systems with unknown inputs, necessitating the additional kernel inclusion condition \eqref{kernel inclusion} to ensure disturbance decoupling.
\end{remark}

To establish the equivalence between the proposed data-driven approach and the traditional model-based one, we impose the following structural matching condition, which is a generalization of the dual normalizability condition in the previous section.


\begin{assumption}\label{assump_obs_UIO}\begin{equation}\label{eq:rank_condition_matching}
		rk\begin{bmatrix}E & F\\ C & 0\end{bmatrix} = n+q.
	\end{equation}
\end{assumption}
\begin{remark}
	Condition \eqref{eq:rank_condition_matching}, known as the observer matching condition, guarantees instantaneous unknown input isolation via the measurement $y(k)$ and the algebraic constraints of $E$ \cite{zhang2017discrete,darouach1996reduced}.  When the descriptor system \eqref{3_1} reduces to a standard LTI system (i.e., $E = I_n$), condition \eqref{eq:rank_condition_matching} simplifies to $rk(CF) = rk(F)$, which exactly coincides with the classical UIO matching condition (see~\cite{darouach1994full}). 
\end{remark}
Under Assumption \ref{assump_obs_UIO}, since 
$rk{\tiny\begin{bmatrix}E & F\\ C & 0 \\I_n & 0\end{bmatrix}}=n+rk(F)=rk{\tiny\begin{bmatrix}E & F\\ C & 0\end{bmatrix}},$
there exist matrices $\bar{T}\in\mathbb{R}^{n\times n}$ and $\bar{N}\in\mathbb{R}^{n\times p}$ such that
\begin{equation}\label{eq:uio_TN_condition}
	\bar{T}E+\bar{N}C=I_n,\quad \bar{T}F=0.
\end{equation}
By using \eqref{eq:uio_TN_condition}, the system \eqref{3_1} is transformed into:
\begin{equation}\label{eq:transformed_system_uio}
	x(k+1) = \bar{T}Ax(k) + \bar{T}Bu(k) + \bar{N}y(k+1).
\end{equation}

\begin{theorem}\label{thm:equivalence_uio}
	\textbf{\textup{(Equivalence of model-based and data-driven UIO conditions)}}
	Suppose that Assumptions~\ref{assump:PE_uio} and \ref{assump_obs_UIO} hold. Then, the kernel inclusion condition \eqref{kernel inclusion} is always satisfied. Moreover, the following statements are equivalent:\\
	(i) There exists a full-order UIO of the form \eqref{full_order_observer} for system \eqref{3_1}.\\
	(ii) There exist parameter matrices $\bar{\Sigma}_{X_p}$, $\bar{\Sigma}_{U_p}$, $\bar{\Sigma}_{Y_p}$, and $\bar{\Sigma}_{Y_f}$ satisfying \eqref{eq:his_into_observer2}, with $\bar{\Sigma}_{X_p}$ being Schur stable.\\
	(iii) \begin{equation}\label{rank_equ_uio}
		rk\begin{bmatrix}
			\lambda X_p-X_f\\U_p\\Y_p\\Y_f
		\end{bmatrix} = rk\begin{bmatrix}
			X_p\\U_p\\Y_f
		\end{bmatrix},
		\forall \lambda \in \mathbb{C},\,\  |\lambda| \geq 1.
	\end{equation}   
	(iv)
	\begin{equation}\label{eq:rank_condition_detectability}
		rk\begin{bmatrix}
			\lambda E - A & -F\\
			C & 0
		\end{bmatrix}
		= n+q,\,\ \forall\,\lambda\in\mathbb C,\,\ |\lambda|\ge 1.
	\end{equation}
\end{theorem}
\begin{proof}
	Under Assumption \ref{assump_obs_UIO}, the historical data $\{u_d,x_d,y_d\}$ collected from system \eqref{3_1} should satisfy \eqref{eq:transformed_system_uio}, leading to 
	$$X_f=\bar TAX_p+ \bar{T}BU_p + \bar{N}Y_f.$$ Hence, the kernel inclusion condition \eqref{kernel inclusion} holds. 
	
	$(i)\Leftrightarrow(ii)$: This has been proved in Theorem \ref{thm:full_order_data-driven_UIO}.
	
	$(ii)\Rightarrow(iii)$:
	To prove $(iii)$, first notice that
	\[
	\begin{aligned}
		rk\left[\begin{array}{l}X_p \\ U_p \\ Y_p \\ Y_f\end{array}\right]&=rk\left(\left[\begin{array}{ccc}I & 0 & 0 \\ 0 & I & 0 \\ C & 0 & 0 \\ 0 & 0 & I\end{array}\right]\left[\begin{array}{l}X_p \\ U_p \\ Y_f\end{array}\right]\right)
		= rk\left[\begin{array}{l}X_p \\ U_p \\ Y_f\end{array}\right]. 
	\end{aligned}
	\]
	By exploiting condition $(ii)$, we obtain
	\[\begin{aligned}
		\scalebox{0.8}{
			$rk\left[\begin{array}{c}\lambda X_p-X_f \\ U_p \\ Y_p \\ Y_f\end{array}\right] = rk\left(\left[\begin{array}{cccc}\lambda I-\bar{\Sigma}_{X_p} & -\bar{\Sigma}_{U_p} & -\bar{\Sigma}_{Y_p} & -\bar{\Sigma}_{Y_f} \\ 0 & I & 0 & 0 \\ 0 & 0 & I & 0 \\ 0 & 0 & 0 & I\end{array}\right]\left[\begin{array}{l}X_p \\ U_p \\ Y_p \\ Y_f\end{array}\right]\right).$}
	\end{aligned}\]
	Since $\bar{\Sigma}_{X_p}$ is Schur, it follows that the rank of the matrix on the left coincides with rank $[ X_p^\top, U_p^\top, Y_p^\top, Y_f^\top]^\top$ for every $\lambda\in\mathbb{C}$ with $|\lambda|\ge 1$.  Therefore, condition $(iii)$ can be derived directly from condition $(ii)$. 
	
	$(iii)\Rightarrow(iv)$:
	Expanding the left-hand data matrix in \eqref{rank_equ_uio} yields:
	\begin{align}
		&rk \begin{bmatrix} \lambda X_p - X_f \\ U_p \\ Y_p\\Y_f \end{bmatrix} = rk \left( \begin{bmatrix} I & 0 & 0 \\ 0 & I & 0 \\ 0 & 0 & I \\ -C & 0 & \lambda I \end{bmatrix} \begin{bmatrix} \lambda X_p - X_f \\ U_p \\ Y_p \end{bmatrix} \right) \notag \\
		&= \begin{bmatrix} \lambda X_p - X_f \\ U_p \\ Y_p \end{bmatrix} = rk \begin{bmatrix} \lambda Z_{1p} - Z_{1f} \\ \lambda Z_{2p} - Z_{2f} \\ U_p \\ Y_p \end{bmatrix} \notag \\
		&= rk \scalebox{0.75}{$\begin{bmatrix} 
				\lambda I - A_1 & -B_1 & -F_1 & 0 & \dots & 0 & 0 \\
				0 & -\lambda B_2 & -\lambda F_2 & -(\lambda R-I)B_2 & \dots & -R^{s-1}B_2 & -R^{s-1}F_2 \\
				0 & I_m & 0 & 0 & \dots & 0 & 0 \\
				C_1 & -C_2B_2 & -C_2F_2 & -C_2RB_2 & \dots & -C_2R^{s-1}B_2 & -C_2R^{s-1}F_2
			\end{bmatrix}$} \notag \\
		&= m + rk \begin{bmatrix} 
			\lambda I - A_1 & -F_1 & 0 \\ 
			0 & -\lambda F_2 & -(\lambda R - I)Q' \label{eq_rank_matrix_pencil} \\ 
			C_1 & -C_2F_2 & -C_2RQ' 
		\end{bmatrix},
	\end{align}
	where $Q'\doteq \begin{bmatrix} B_2 & F_2 & RB_2 & RF_2 & \dots & R^{s-1}B_2 & R^{s-1}F_2 \end{bmatrix}$.
	
	Define $\Omega\!=\!\scalebox{0.7}{$\begin{bmatrix}
			I_{n_1} & 0 & 0 & 0 & \dots & 0 & 0 \\
			0 & -B_2 & -F_2 & -RB_2 & \dots & -R^{s-1}B_2 & -R^{s-1}F_2 \\
			0 & I_m & 0 & 0 & \dots & 0 & 0 \\
			C_1A_1 & C_1B_1 & C_1F_1 & -C_2B_2 & \dots & -C_2R^{s-2}B_2 & -C_2R^{s-2}F_2
		\end{bmatrix}$}$. Note that (c.f. \eqref{eq:historical_data})
	\begin{equation} \label{eq_data_uio}
		[X_p^\top,U_p^\top,Y_f^\top]^\top=\Omega \begin{bmatrix} H_1(z_{1d,[0,T-1]}) \\ H_{s+1}(\tilde{u}_{d,[0,T+s-1]}) \end{bmatrix}
	\end{equation}
	Expanding the right-hand data matrix in \eqref{rank_equ} gives
	\begin{align}
		&rk \begin{bmatrix} X_p \\ U_p \\ Y_f \end{bmatrix} 
		= rk \begin{bmatrix} Z_{1p} \\ Z_{2p} \\ U_p \\ Y_f \end{bmatrix}=rk(\Omega)\notag \\
		&= n_1 + m + rk \left( \begin{bmatrix} R & F_2 \\ C_2 & -C_1F_1 \end{bmatrix} \begin{bmatrix} 0 & -Q' \\ -I_q & 0 \end{bmatrix} \right), \label{eq-corollary}
	\end{align}where the second equality is due to \eqref{eq_data_uio} and Assumption \ref{assump:PE_uio}.
	Note also that
	\begin{equation}
		\begin{aligned} \label{eq:matching}
			rk \begin{bmatrix} E & F \\ C & 0 \end{bmatrix} &= rk \left( \begin{bmatrix} 
				I & 0 & 0 \\ 
				0 & I & 0 \\ 
				-C_1 & 0 & I 
			\end{bmatrix} 
			\begin{bmatrix} 
				I_{n_1} & 0 & F_1 \\ 
				0 & R & F_2 \\ 
				C_1 & C_2 & 0 
			\end{bmatrix} \right) \\
			&= rk \begin{bmatrix} 
				I_{n_1} & 0 & F_1 \\ 
				0 & R & F_2 \\ 
				0 & C_2 & -C_1F_1 
			\end{bmatrix} \\
			&= n_1 + rk \begin{bmatrix} R & F_2 \\ C_2 & -C_1F_1 \end{bmatrix}
		\end{aligned}.
	\end{equation}   Since $rk \begin{bmatrix} E & F \\ C & 0 \end{bmatrix} = n + q$, it follows that $rk \begin{bmatrix} R & F_2 \\ C_2 & -C_1F_1 \end{bmatrix} = n_2 + q$, yielding
	\begin{equation}\label{rank_right}
		rk([X_p^\top,U_p^\top,Y_f^\top]^\top) = n_1 + m + q + rk(Q') .
	\end{equation}
	
	Pre-multiplying by a non-singular block lower-triangular transformation matrix and utilizing the algebraic identity $\lambda C_2R(\lambda R-I)^{-1} = C_2(I + (\lambda R-I)^{-1})$ leads to
	\begin{align}
		&\scalebox{0.8}{$\begin{bmatrix} I_{n_1} & 0 & 0 \\ 0 & I_{n_2} & 0 \\ 0 & -C_2R(\lambda R-I)^{-1} & I_p \end{bmatrix}
			\begin{bmatrix} \lambda I - A_1 & -F_1 & 0 \\ 0 & -\lambda F_2 & -(\lambda R - I)Q' \\ C_1 & -C_2F_2 & -C_2RQ' \end{bmatrix}$} \notag \\
		&= \scalebox{0.8}{$\begin{bmatrix} \lambda I - A_1 & -F_1 & 0 \\ 0 & -\lambda F_2 & -(\lambda R - I)Q' \\ C_1 & C_2(\lambda R-I)^{-1}F_2 & 0 \end{bmatrix}$}. \label{eq_block_equality}
	\end{align}
	Note that $Im(\lambda RQ')\subseteq Im(Q')$ and therefore $Im ((\lambda R-I)Q')\subseteq Im(Q')$.   Since $\lambda R-I$ is non-singular for $\lambda\in {\mathbb C}$ with $|\lambda|>1$, $rk((\lambda R-I)Q')=rk(Q')$.  Therefore, $Im ((\lambda R-I)Q')= Im(Q')$.
	As $Im(\lambda F_2)\subseteq Im(Q')$, from \eqref{eq_rank_matrix_pencil} and \eqref{eq_block_equality} we have 
	\begin{equation} \label{eq:rank_ineq}
		rk \scalebox{0.8}{$\begin{bmatrix} \lambda X_p - X_f \\ U_p \\ Y_p\\Y_f \end{bmatrix}$}  = m + rk \scalebox{0.8}{$\begin{bmatrix} \lambda I - A_1 & -F_1 \\ C_1 & C_2(\lambda R-I)^{-1}F_2 \end{bmatrix}$} + rk(Q').
	\end{equation}
	By comparing \eqref{rank_right} and \eqref{eq:rank_ineq}, condition (iii) implies
	$rk \begin{bmatrix} \lambda I - A_1 & -F_1 \\ C_1 & C_2(\lambda R-I)^{-1}F_2 \end{bmatrix} = n_1 + q$.    
	
	Define non-singular block transformation matrices $\mathcal{T}_L$ and $\mathcal{T}_R$ as
	{\small\begin{equation*}
			\mathcal{T}_L = \scalebox{0.9}{$\begin{bmatrix} I_{n_1} & 0 & 0 \\ 0 & -C_2(\lambda R-I)^{-1} & I_p \\ 0 & I_{n_2} & 0 \end{bmatrix}$}, \ 
			\mathcal{T}_R = \scalebox{0.9}{$\begin{bmatrix} I_{n_1} & 0 & 0 \\ 0 & (\lambda R-I)^{-1}F_2 & I_{n_2} \\ 0 & I_q & 0 \end{bmatrix}$}.
	\end{equation*}}
	Pre-multiplying and post-multiplying by $\mathcal{T}_L$ and $\mathcal{T}_R$ gives
	\begin{align*}
		&rk\begin{bmatrix} \lambda E - A & -F \\ C & 0 \end{bmatrix}= rk\begin{bmatrix} \lambda I - A_1 & 0 & -F_1 \\ 0 & \lambda R - I & -F_2 \\ C_1 & C_2 & 0 \end{bmatrix}\\
		&\quad = rk\left( \mathcal{T}_L 
		\begin{bmatrix} \lambda I - A_1 & 0 & -F_1 \\ 0 & \lambda R - I & -F_2 \\ C_1 & C_2 & 0 \end{bmatrix} 
		\mathcal{T}_R \right)\\
		&\quad = rk\begin{bmatrix} \lambda I - A_1 & -F_1 & 0 \\ C_1 & C_2(\lambda R-I)^{-1}F_2 & 0 \\ 0 & 0 & \lambda R - I \end{bmatrix} \\
		&\quad = rk\begin{bmatrix} \lambda I - A_1 & -F_1 \\ C_1 & C_2(\lambda R-I)^{-1}F_2 \end{bmatrix} + rk(\lambda R - I)\\
		&\quad =n_1+q+rk(\lambda R-I).
	\end{align*}
	Considering that $R$ is a nilpotent matrix (which renders $\lambda R - I$ non-singular), it is conclusively determined that \eqref{eq:rank_condition_detectability} holds.

	$(iv)\Rightarrow(i)$:
	If condition $(iv)$ is satisfied, the existence of a UIO with the structure defined in (\ref{full_order_observer}) is guaranteed by Theorem 1 in \cite{gupta2015full}.
\end{proof}
\begin{remark}\label{rem:degeneration_to_LTI}
	Consider the special case $E=I_n$, where the descriptor system \eqref{3_1} reduces to a normal LTI system. Building upon the simplification of the matching condition discussed earlier, the detectability condition \eqref{eq:rank_condition_detectability} naturally reduces to 
	$rk\begin{bmatrix} \lambda I - A & -F \\ C & 0 \end{bmatrix} = n + q, \quad \forall \lambda\in\mathbb{C}, \,\ |\lambda| \ge 1$. 
	This aligns perfectly with the classical UIO existence condition for LTI systems \cite{disaro2024equivalence}. 
\end{remark}

We end this section by providing data-based conditions under which the core Assumption \ref{assump_obs_UIO} is satisfied.

\begin{corollary}[Data-based test for Assumption \ref{assump_obs_UIO}]
	Assumption \ref{assump_obs_UIO} holds if 
	$$rk([X_p^\top,U_p^\top,Y_f^\top]^\top)=n+m+q.$$Moreover, if system \eqref{3_1} is R-controllable with the augmented input matrix $[B,F]$ and Assumption \ref{assump:PE_uio} holds, then this condition is also necessary for the satisfaction of Assumption \ref{assump_obs_UIO}.
\end{corollary}

\begin{proof}
	We first show the sufficiency. From \eqref{eq-corollary}, to make
	$rk([X_p^\top,U_p^\top,Y_f^\top]^\top)=n+m+q$, it must hold that 
	$rk \begin{bmatrix} R & F_2 \\ C_2 & -C_1F_1 \end{bmatrix} = n_2 + q$. Then from \eqref{eq:matching}, condition \eqref{eq:rank_condition_matching} holds. Next, under the proposed additional conditions, $\begin{bmatrix} 0 & -Q' \\ -I_q & 0 \end{bmatrix}$ has full row rank. Then, from \eqref{eq-corollary} we have
	$rk([X_p^\top,U_p^\top,Y_f^\top]^\top)=n_1+m+rk(\begin{bmatrix} R & F_2 \\ C_2 & -C_1F_1 \end{bmatrix})$. From \eqref{eq:matching}, we obtain that $rk([X_p^\top,U_p^\top,Y_f^\top]^\top)=n+m+q$ is necessary and sufficient for condition \eqref{eq:rank_condition_matching}.
\end{proof}

\section{Application to extended state observers}\label{sec: eso}
Extended state observers (ESO) play an important role in active disturbance
rejection control (ADRC) \cite{han2009pid}.  Unlike UIOs, ESO estimates not only state variables but also disturbances.
Here, we show that our data-based observers for descriptor systems could be directly applied to the design of ESO. 

For discrete-time linear systems, consider the following model:
\begin{equation}\label{before extended system}
	\begin{cases}
		x(k+1) = A_0 x(k) + B_0 u(k) + E_0 d(k),\\
		y(k) = C_0 x(k) + F_0d(k),
	\end{cases}
\end{equation}
where $x(k) \in \mathbb{R}^n$ is the state vector, $u(k) \in \mathbb{R}^m$ is the known input, $d(k) \in \mathbb{R}^r$ is the disturbance, and $y(k) \in \mathbb{R}^p$ is the output. The matrices $A_0$, $B_0$, $E_0$, $C_0$, and $F_0$ are of appropriate dimensions, with $C_0$ typically assumed to have full row rank and $F_0$ assumed to have full column rank. Without loss of generality, assume that
$[E_0^\top, F_0^\top]^\top$ has full column rank.

The key to ESO design is that the disturbance $d(k)$ is treated as an additional state.
The augmented system can be written as
\begin{equation}\label{extended system}
	\begin{cases}
		E\tilde x(k+1)=A\tilde x(k)+Bu(k),\\
		y(k)=C\tilde x(k),
	\end{cases}
\end{equation}
where $\tilde x(k) = \begin{bmatrix}
	x(k)\\d(k)
\end{bmatrix}$, $E = \begin{bmatrix}
	I_n & 0\\ 0 & 0
\end{bmatrix}$, $A=\begin{bmatrix}
	A_0 & E_0\\ 0 & 0
\end{bmatrix}$, $B = \begin{bmatrix}
	B_0\\0
\end{bmatrix}$, and $C = \begin{bmatrix}
	C_0 & F_0
\end{bmatrix}$.

It is not difficult to see that the augmented system \eqref{extended system} satisfies Assumption \ref{assump_obs}, i.e. $rk\begin{bmatrix}
	E\\C
\end{bmatrix}=rk\begin{bmatrix}
	I_n & 0\\ 0 & 0\\ C_0 & F_0
\end{bmatrix}=n+r$.

A standard form for the discrete-time ESO for the augmented system \eqref{extended system} is given by the following observer dynamics equation
\begin{equation*}
	\hat{x}(k+1) = A_O\hat{x}(k) + B_O^uu(k) + B_O^yy(k) + N_Oy(k+1).
\end{equation*}
Since Assumption \ref{assump_obs} holds, the necessary and sufficient condition for the existence of a model-based ESO (or similar variants) is the detectability of the augmented system \eqref{extended system}, i.e., $rk \begin{bmatrix} \lambda E - A \\ C \end{bmatrix} =n+r,\,\ \forall \lambda \in \mathbb{C}, \ |\lambda| \geq 1$.
This condition can be reformulated by expressing the matrices of the augmented system in terms of the original system matrices \( (A_0, E_0, C_0, F_0) \). Specifically, the rank condition becomes equivalent to   
\[
rk \begin{bmatrix} \lambda I_n - A_0 & -E_0 \\ C_0 & F_0 \end{bmatrix} = n + r,\,\ \forall \lambda \in \mathbb{C}, \, |\lambda|
\geq 1, \, \lambda \text{ finite}.
\]
This reformulation explicitly links the detectability of the augmented system \eqref{extended system} to the strong detectability of the original system \eqref{before extended system} (the triple \( (A_0, E_0, C_0, F_0) \) is strongly detectable if and only if the above rank condition holds). This model-based existence condition coincides with another augmented system based ESO design method \cite{chen2025model}, which uses the difference between the future disturbance and the current one, i.e., $d(k+1)-d(k)$, as an external input. 
Our proposed method provides a data-driven design of ESO for system \eqref{before extended system}. 
The data-driven construction of an ESO fundamentally relies on obtaining input-output-state data sequences that can sufficiently excite the dynamic modes of system \eqref{extended system}. As analyzed in Section \ref{sec: full order}, for the ESO to be constructed, its historical input data $\{u_d(k)\}_{k=0}^{T+s-1}$ must first satisfy the persistent excitation condition in Assumption \ref{assump:PE}. Under this premise, the design further requires the historical input-output data $(\{u_d(k)\}_{k=0}^{T-1},\{y_d(k)\}_{k=0}^{T-1})$, and the augmented state vector \( \{\tilde x_d(k)\}_{k=0}^{T-1} \), which is defined by concatenating the original system state and the disturbance input. Note that without the disturbance input data, it is impossible to estimate $d(k)$ online in the observer framework, since any invertible coordinate transformation does not alter the system input and output.  

	

The complete augmented state data can be acquired during specific experimental phases, system commissioning, or through high-fidelity sensing and estimation techniques employed temporarily for observer calibration. The specific observer construction procedure is detailed in Theorem \ref{thm:data_driven observer}, which provides a data-driven characterization for the existence of an ESO, specifying the necessary and sufficient conditions that must be satisfied. Note that UIO is not appliable here since an UIO does not estimate the unknown input, while the distrubance must be inferred in an ESO. 

\begin{remark}
While collecting historical disturbance data is feasible during debugging or specific experimental phases, relying on continuous online monitoring for daily operations is often impractical. Maintaining a dedicated sensor network not only increases hardware costs but also faces implementation hurdles in complex industrial environments, such as spatial constraints or the need for invasive modifications. Alternatively, an ESO acts as a `soft sensor', estimating disturbances in real-time using only standard control inputs and outputs. This approach circumvents the need for additional instrumentation, offering a more practical and cost-effective solution.
\end{remark}


\section{Numerical Examples}\label{sec: simulation}
In this section, we present some numerical examples to validate the efficacy of the proposed observers.
\begin{example}\label{example1}
	Consider the discrete-time descriptor system in the form of \eqref{2_1} with
	\begin{align*}
		E &= \begin{bmatrix}
			1 & 2 & 1\\
			0 & 2 & 1\\
			1 & 0 & 0
		\end{bmatrix},\quad
		A = \begin{bmatrix}
			0.153 & 0.045 & 0.069\\
			0.156 & 0.252 & 0.156\\
			0.135 & -0.171 & -0.636
		\end{bmatrix},\\
		B &= \begin{bmatrix}
			1\\
			1\\
			0.2
		\end{bmatrix},\quad
		C = \begin{bmatrix}
			1 & 0 & 0\\
			1 & 0 & 1
		\end{bmatrix}.
	\end{align*}
	A data-driven observer is constructed according to Theorem~\ref{thm:data_driven observer}. The state data length is chosen as $T=20$. The historical input sequence has length $T+s-1=20$ and is generated from the uniform distribution on $(-5,5)$. The initial slow-subsystem state is selected randomly, and the free parameter in Remark~\ref{remark3} is chosen as $K_1=0$. The resulting observer matrix $\Sigma_{X_p}$ is
	\begin{align*}
		\Sigma_{X_p}=
		\begin{bmatrix}
			0.0046 & -0.2001 & -0.0257\\
			0.0397 & 0.0037 & 0.1422\\
			-0.0046 & 0.2001 & 0.0257
		\end{bmatrix},
	\end{align*}
	whose eigenvalues have maximum modulus $0.2083$. Hence, $\Sigma_{X_p}$ is Schur stable. For validation, another trajectory is generated with testing input $u(k)=4\sin(k)$. The initial slow-subsystem state and the initial observer state $\hat{x}(0)$ are both selected randomly from $(0,2)^2$ and $(0,2)^3$, respectively. Figure~\ref{fig:image1} shows that the estimated states closely track the true states.
	
	\begin{figure}[H]
		\centering
		\includegraphics[ width=0.48\textwidth]{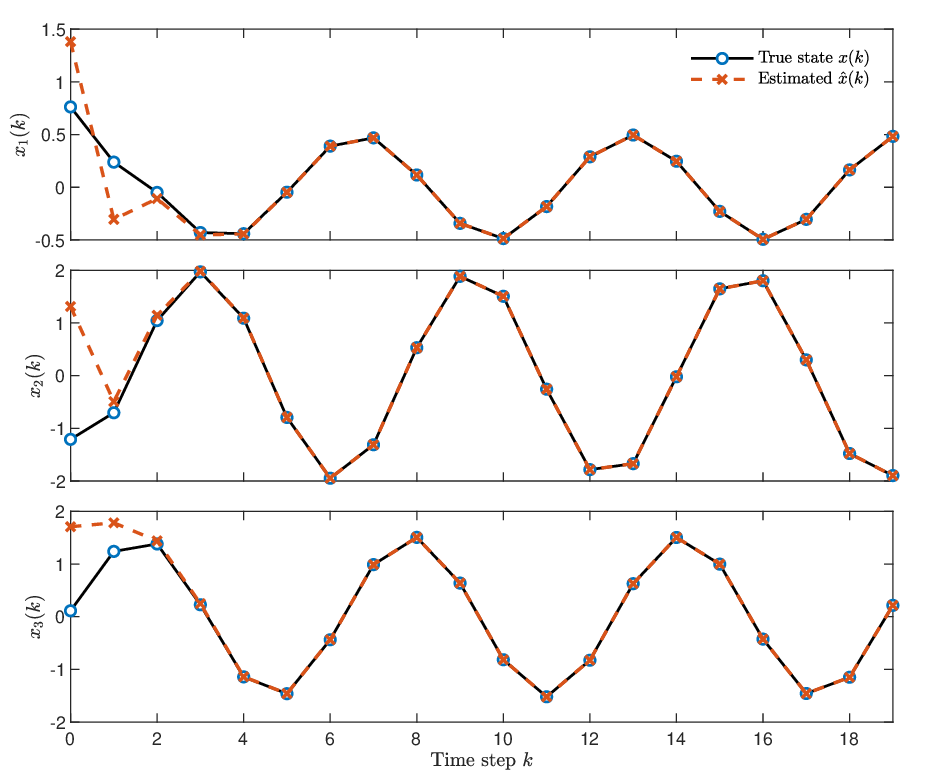}
		\caption{State estimation of the data-driven   observer.}
		\label{fig:image1}
	\end{figure}
\end{example}

\begin{example}\label{example2}
	Consider the descriptor system with unknown inputs in the form of \eqref{3_1}, where the matrices $E$, $A$, $B$, and $C$ are the same as in Example~\ref{example1}, and
	\begin{align*}
		F=[1,\ 0.2,\ 0.5]^\top.
	\end{align*}
	A data-driven UIO is then constructed according to Theorem~\ref{thm:full_order_data-driven_UIO}. The state data length is again chosen as $T=20$. The historical known input and unknown input both have length $T+s-1=20$ and are independently generated from the uniform distribution on $(-5,5)$. The resulting observer matrix $\bar{\Sigma}_{X_p}$ is
	\begin{align*}
		\bar{\Sigma}_{X_p}=
		\begin{bmatrix}
			0.1501 & -0.0334 & -0.3300\\
			0.0743 & 0.1343 & -0.0961\\
			-0.1501 & 0.0334 & 0.3300
		\end{bmatrix},
	\end{align*}
	whose eigenvalues have maximum modulus $0.4628$. Hence, $\bar{\Sigma}_{X_p}$ is Schur stable. For validation, the known input is chosen as $u(k)=4\sin(k)$, while the unknown input is generated randomly from $(-1,1)$. The initial slow-subsystem state and the initial observer state are selected in the same manner as in Example~\ref{example1}. Figure~\ref{fig:image3} shows accurate state reconstruction in the presence of unknown inputs.
	
	\begin{figure}[H]
		\centering
		\includegraphics[ width=0.48\textwidth]{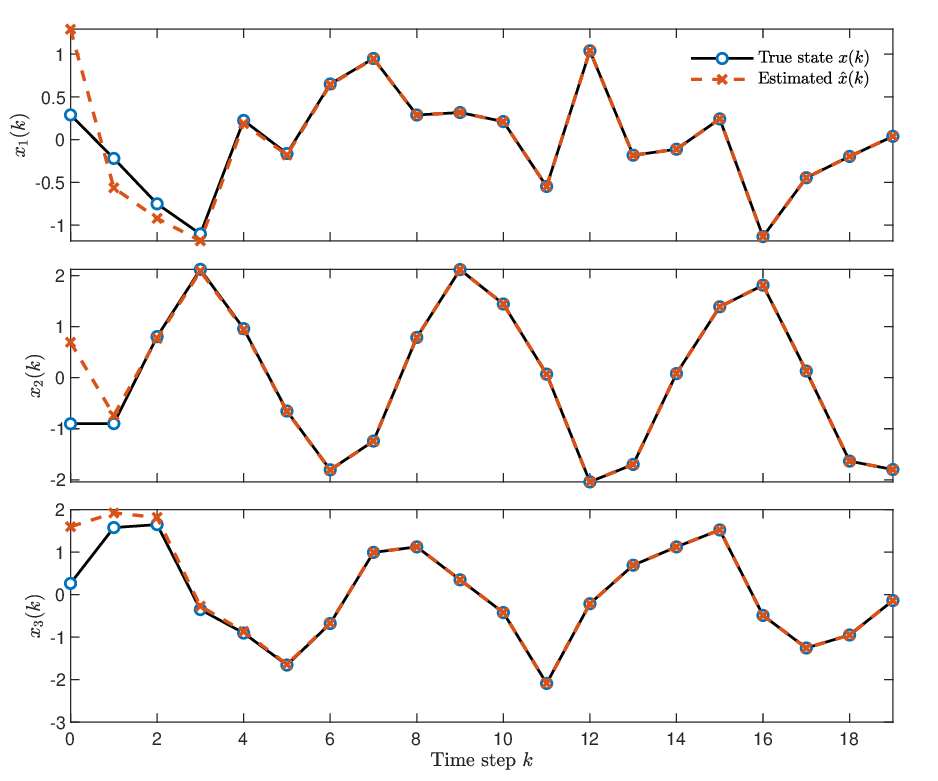}
		\caption{State estimation of the data-driven   UIO.}
		\label{fig:image3}
	\end{figure}
\end{example}

\begin{example}\label{example4}
	Consider the linear time-invariant system in the form of \eqref{before extended system} with
	\begin{align*}
		A_0&=\begin{bmatrix}
			0.153 & 0.045 & 0.069\\
			0.156 & 0.252 & 0.156\\
			0.135 & -0.171 & -0.636
		\end{bmatrix},\quad
		B_0=\begin{bmatrix}
			1\\
			1\\
			0.2
		\end{bmatrix},\\
		E_0&=\begin{bmatrix}
			1 & 0\\
			0 & 0\\
			0 & 1
		\end{bmatrix},\quad
		C_0=\begin{bmatrix}
			1 & 0 & 0\\
			0 & 1 & 0
		\end{bmatrix},\quad
		F_0=\begin{bmatrix}
			1.2 & 1\\
			0 & 1
		\end{bmatrix}.
	\end{align*}
	The augmented state is $\tilde x(k)=\begin{bmatrix}x(k)^\top & d(k)^\top\end{bmatrix}^\top\in\mathbb{R}^5$. A data-driven ESO is constructed according to Theorem~\ref{thm:data_driven observer}. The augmented state data length is chosen as $T=25$. The historical known input sequence has length $T+s-1=25$ and is generated from the uniform distribution on $(-5,5)$, while the disturbance sequence has the same length and is generated from the uniform distribution on $(-3,3)$. The initial plant state $x(0)$ is selected randomly. The resulting observer matrix $\Sigma_{X_p}$ is
	\begin{align*}
		\Sigma_{X_p}=
		\begin{bmatrix}
			-0.1728 & 0.1386 & 0.0690 & 0.6090 & -0.2322\\
			0.1385 & 0.1738 & 0.1560 & -0.0210 & -0.0957\\
			-0.0741 & -0.3776 & -0.6360 & -0.2509 & 0.5843\\
			0.2594 & 0.0294 & 0.0725 & -0.5251 & 0.1138\\
			-0.1385 & -0.1738 & -0.1560 & 0.0210 & 0.0957
		\end{bmatrix},
	\end{align*}
	whose eigenvalues have maximum modulus $0.7426$. Hence, $\Sigma_{X_p}$ is Schur stable. For validation, the known input is chosen as $u(k)=4\sin(k)$, the disturbance sequence is generated randomly from $(-2,2)^2$, the initial plant state is selected randomly from $(-2,0)^3$, and the initial observer state $\tilde x(0)$ is selected randomly from $(0,2)^5$. Figure~\ref{fig:image9} shows accurate simultaneous estimation of the plant states and disturbances.
	
	\begin{figure}[t]
		\centering
		\includegraphics[width=0.48\textwidth]{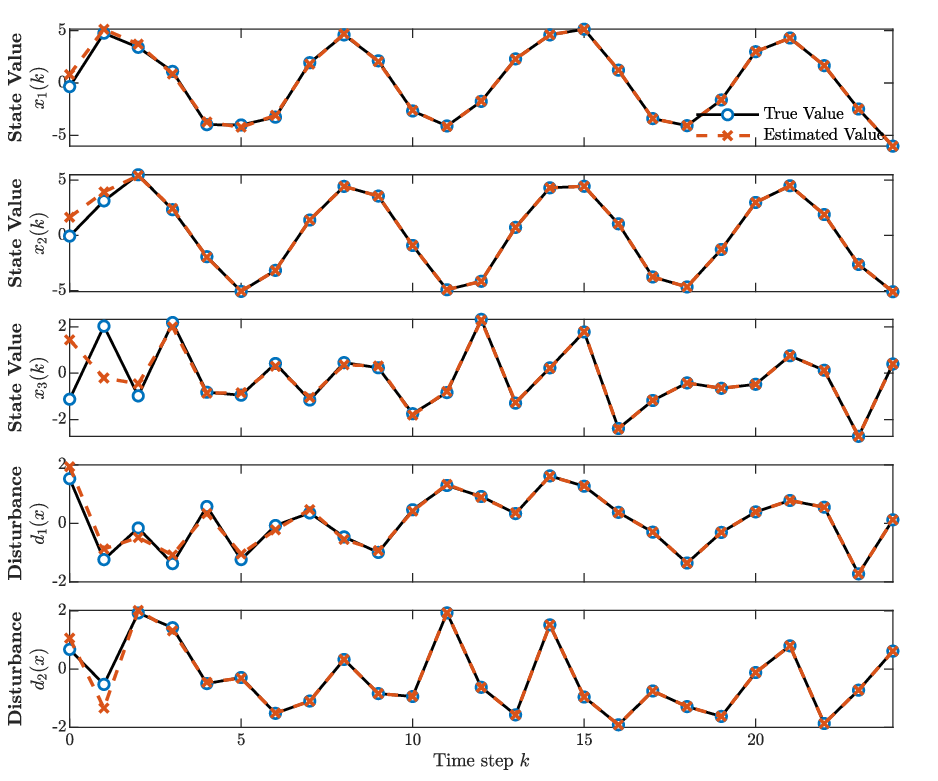}
		\caption{State estimation of the data-driven ESO.}
		\label{fig:image9}
	\end{figure}
\end{example}
\section{Conclusion}
This paper has established a data-driven framework for state observer designs in descriptor systems. By utilizing historical data, we derived data-based existence conditions and design for a normal state observer. Furthermore, the framework was extended to provide a data-driven design for full-order UIO to effectively handle unknown inputs. Notably, for both standard state observers and UIOs, we rigorously proved the mathematical equivalence between the proposed data-driven and classical model-based approaches. Additionally, the versatility of the proposed approach was demonstrated through its application to ESO design, enabling simultaneous estimation of system states and disturbances via system augmentation. Future work will focus on designing reduced-order observers, enhancing the robustness of the proposed method against measurement noise, and extending the current linear framework to nonlinear descriptor systems.

\appendix

\bibliographystyle{elsarticle-num}        
{\footnotesize
	\bibliography{example}
}
\end{document}